\documentclass[12pt, reqno]{amsart}

 \numberwithin{equation}{section}
\usepackage{amsfonts}
\usepackage{amscd}
\usepackage{amssymb}
\usepackage{graphics}
\usepackage{amsmath}
\usepackage[english]{babel}
\usepackage{hyperref}
\usepackage{bbm}
\usepackage{yfonts}
\usepackage{mathrsfs}
\usepackage{color}
\usepackage{graphicx}

\newtheorem{theorem}{Theorem}
\newtheorem{claim}[theorem]{Claim}
\newtheorem{conclusion}[theorem]{Conclusion}
\newtheorem{conjecture}[theorem]{Conjecture}

\newtheorem{lemma}[theorem]{Lemma}

\newtheorem{proposition}[theorem]{Proposition}

\newtheorem{prop}{Proposition}[section]

\newtheorem{obs}[prop]{Observation}

\newtheorem{defn}[prop]{Definition}
\newtheorem{rem}[prop]{Remark}

\newtheorem{exam}[prop]{Example}

\newtheorem{rmk}[prop]{Remark}

\newcommand{\BN}{{\mathbb{N}}}
\newcommand{\BR}{{\mathbb{R}}}

\newcommand{\calF}{{\mathcal{F}}}

\newcommand{\CA}{{\mathcal{A}}}

\begin{document}
\title{Geometry and Dynamics in Zero Temperature Statistical Mechanics Models}
\author{Ran J. Tessler}
\maketitle
\tableofcontents
\newpage
To Yaki and Sara, my grandparents, and to Boris.
\newpage
\begin{abstract}

We consider several models whose motivation arises from statistical mechanics.
We begin by investigating some families of distributions of translation
invariant subgraphs of some Cayley graphs, and in particular subgraphs of the square lattice.
We then discuss some properties of the Spin-Glass model in that lattice. We
continue in describing some properties of the Spin-Glass models in some other
graphs.

The last two parts of this work are devoted to the understanding of two
dynamical processes on graphs. The first one is the well known zero-temperature Glauber dynamics on some families of graphs.

The second dynamics, which we call the Loop Dynamics, is a natural generalization of the zero-temperature Glauber dynamics, which appears to have some interesting properties. We analyzed some of its properties for planar lattices, though the exact same techniques are applied for larger families of graphs as well.
\end{abstract}

%\newpage
%\tableofcontents
%\newpage

% ********************************************************************************************************************************
% ********************************** Section I - The Model ***********************************************************************
% ********************************************************************************************************************************
\newpage
\section{Introduction}
\subsection{Background}
$~$\\
This paper is divided into four main theoretical sections. The remaining parts of this paper are the introduction, the Hebrew abstract and most importantly the dedications. Each of the main theoretical sections contains subsections of background, notations and formulation of results, while the rest of each section is devoted to proving these results.

The first section deals with properties of distributions of translation invariant families of subgraphs of amenable Cayley graphs. In particular we investigate some properties of such distributions in the square lattice. These results, in addition to some independent interest are used in proofs of some results from later sections.

The Second part considers the Edwards-Anderson Ising spin glass model on the planar square lattice and other families of graphs. For the square lattice spin glass we obtain new results about the geometry of its ground states (to be defined below).

We then continue in exploring the Edwards-Anderson Ising spin glass model on other graphs. We show that for wide families of graphs we have exactly one ground state, while for a large family of other graphs we have an infinite amount.

The last two part parts are devoted to exploring two dynamical processes. We first consider zero-temperature Glauber dynamics on some families of graphs. Our main result in that section is  that on regular trees with an even degree this dynamics does not freeze. This is a complement to a known result that for regular trees with an odd degree the dynamics does freeze.

We then define and discuss a new dynamics $\emph{the Loop dynamics}$ which is a natural generalization of the zero-temperature Glauber dynamics. It is more difficult to show that this process is well defined, but it appears to be worth the effort, as weak limits of this dynamical process happen to have some interesting properties.

The first two parts and the forth part are taken from a joint work with Noam Berger. The third is a joint work with Oren Louidor.

%-------------------------------------------------------------------------------------------------------

\subsection{Ergodic Theory Background and the Mass Transport Principle}
\begin{defn}
Let $(X,\mathfrak{B},\mu)$ be a finite measure space, and $T:X\rightarrow X$ a measurable map.
\begin{enumerate}
\item We say that $T$ is $\emph{measure~preserving}$ if for every $B\in\mathfrak{B}$ \[\mu(B)=\mu(T^{-1}(B))\].

\item We say that a measure preserving map $T$ is $\emph{ergodic}$ if the only measurable sets $B\in\mathfrak{B}$ such that
\[ \mu(T^{-1}(B)\triangle B)=0\]
are either sets of measure zero or sets with complement $X\backslash B$ of measure zero.
\end{enumerate}
\end{defn}
\begin{defn}
Let $(X,\mathfrak{B},\mu)$ be a finite measure space, and $G$ a discrete group acting on it, such that $\forall g \in G$ the action of $g$ on $X$ is a measurable map.
\begin{enumerate}
\item We say that $\mu$ is $\emph{invariant}~\emph{under}~\emph{the}~\emph{action}~\emph{of}~G$ if for every $B\in\mathfrak{B}$ and every
$g \in G$  \[\mu(B)=\mu(g\cdot {B})\].

\item We say that the action of $G$ is $\emph{ergodic}$ if the only measurable sets $B\in\mathfrak{B}$ such that
\[ \mu(g\cdot {B}\triangle B)=0\] for every $g\in G$,
are either sets of measure zero or sets with complement $X\backslash B$ of measure zero.
\end{enumerate}
\end{defn}
$~$
\\ $\mathbf{The ~ Mass ~ Transport ~ Principle}$
\\ Let $G=(V,E)$ be a Cayley graph of some discrete group $\Gamma$ or more generally a transitive graph whose automorphism group $\Gamma$ is unimodular. Let $\mu$ be a $\Gamma$-invariant measure on $\{0,1\}^{E}$. A $\emph{mass}$ function $m=m(\omega;x,y)$ is a $\mu$-measurable function from $\{0,1\}^{E}\times{V}\times{V}$ to $\mathbb{R}$. We will also assume that it is invariant under the diagonal action of $\Gamma$, i.e. $m(\omega;x,y)=m(\gamma \omega;\gamma x,\gamma y)$, for any $\gamma\in\Gamma,\omega\in\{0,1\}^{E},x,y\in V$.

A mass function should be thought as the mass transferred from $x$ to $y$ conditioned on the configuration $\omega$.
We denote by $M(x,y)=\mathbb{E}_\mu m(\omega;x,y)$.
The Mass-Transport principle ($\emph{MT}$) says that
\begin{equation}\label{eq:MT}
\forall x\in V ~ \Sigma_{y\in V}M(x,y)=\Sigma_{y\in V}M(y,x)
\end{equation}
In particular the left hand side is finite iff the right hand side is finite.
A good reference for this principle can be found in \cite{lyonsperes}
\newpage
%%-------------------------------------------------------------------------------------------------------
\section{Translation Invariant Measures of Subgraphs in the planar lattice and other graphs}\label{sec:inv subgraphs}
$~$\\

%%-------------------------------------------------------------------------------------------------------
\subsection{Background}\label{sec:Background-inv subgraphs}
$~$\\
This section explores some properties of translation invariant measures of subgraphs of some families of graphs. Our main interest is the planar lattice, where we investigate some families of translation invariant paths and forests and show both qualitative and quantitative results. These objects appear naturally in probability (e.g. \cite{zerner}), graph theory, mathematical physics and more. Some of these results will be used later on in the section on spin glasses.
%%-------------------------------------------------------------------------------------------------------
\subsection{Preliminaries}\label{Prelim-inv subgraphs}
%$~$\\
This subsection is mainly devoted to definitions that of objects that we shall study throughout the section.
\begin{defn}\label{def:direction}
Given a simple, two-sided-infinite path $P=(V(P),E(P))$, a $\emph{direction}~t:V(P)\rightarrow\mathbb{Z}$ is a bijective function from the path's vertices to $\mathbb{Z}$, such that two vertices are neighbors in the path iff their images are consecutive numbers in $\mathbb{Z}$. Given such a direction, and a vertex $v\in V(P)$, we define the $\emph{past}$ of $v$ to be the set of vertices whose $t$-value is smaller than that of $v$, i.e. the set $\{u\in V(P)| t(u)<t(v)\}$. Similarly one defines the $\emph{future}$ of a vertex.

A $\emph{ray}$ is a half-infinite straight line segment which start at some point.
\end{defn}
We have the following trivial observation.
\begin{obs}\label{obs:past is transitive}
Let $u,v,w$ be three vertices of a path with a given direction. Then if $v$ is in the past of $u$, and $w$ is in the past of $v$ then $w$ is in the past of $u$.
\end{obs}
Next we define the notions of a $\emph{cross}$ and a $\emph{snail}$, that use us for obtaining quantitative bounds on biinfinite simple lattice paths which are in the support of some translation invariant measure of paths.
\begin{defn}\label{def:crosses}
Let $P$ be a bi-infinite path in the standard lattice graph $\mathbb{Z}^{2}$. Assume that $P$ intersects any translation, by integer shifts (in both directions),of the positive $x$-axis, negative $x$-axis, positive $y$-axis and negative $y$-axis infinitely many times. Let $p=(a,b)$ be a (lattice) vertex of $P$ and $n$ be a positive integer. We consider the intersection (lattice) points of the path with the half-line $\{(s,b)| s > a\}$. We order these points by their distance from $p$ (or their $x$-coordinate). Denote by ${p_n}^{x_+}$ the $n^{th}$ closest (to $p$) point out of the above intersection points, and similarly define ${p_n}^{x_-}$, ${p_n}^{y_+}$, ${p_n}^{y_-}$. We define the  $\emph{n}^{th}~\emph{cross}~\emph{of}~p$ as the union of two line segment, the one which connects ${p_n}^{x_+}$ and ${p_n}^{x_-}$, and the one which connects ${p_n}^{y_+}$ and ${p_n}^{y_-}$.
We define the $\emph{n}^{th}~\emph{snail}~\emph{of}~p$, as the smallest (with respect to inclusions) connected subgraph of the path which contains $p$ and all the path points from the $n^{th}$-cross of $p$.
The $\emph{length}$ of a snail is defined as the number of path edges included in it.
\end{defn}
We finish this subsection with some definitions concerning infinite trees and forests.
\begin{defn}
Let $T$ be an infinite tree of bounded degree. We say that it is $\emph{single-infinite}$ if it doesn't have two disjoint one-way infinite paths. We say it is $\emph{bi-infinite}$ if it has two disjoint one-way-infinite paths, but not three. Otherwise we say it is $\emph{multi-infinite}$. In a single-infinite tree, any vertex $v$ has a single one-way infinite path which starts at it. We call this path the $\emph{stem}$ of $v$, and denote it by $S(v)$. We denote by $R(v)$, and call it the $emph{roots}$ of $v$, the unique connected component of $T\backslash S(v)$ which contains $v$. Note that if $u$ is a vertex in $R(v)$ then $v$ is a vertex in $S(u)$.

In a bi-infinite tree, there is a single two-way infinite path. We call the $\emph{path}$ of the tree $P(T)$, or simply $P$.

In a multi-infinite tree, a vertex $v$ with the property that $T\backslash \{v\}$ has at least three infinite components is called an $\emph{encounter point}$.

We denote by $\delta_T(x,y)$ the tree distance between the vertices $x,y$. We shall sometimes write only $\delta$, if it is clear to which tree we are referring. The graph distance function will be denoted by $d(\cdot,\cdot)$
\end{defn}
\begin{rmk}\label{rmk:rmk forests}
It follows from the famous K\"{o}nig's Lemma (e.g. \cite{konig}), that any infinite tree of bounded degree contains at least one one-way infinite path. If it has no two such disjoint paths, then unless the tree itself is a one-way infinite path (what is impossible in the translation invariant scenario), it has no unique or canonic one-way infinite path. On the other hand each vertex has a single such path starting from it.

If the tree is bi-infinite is has exactly one bi-infinite path.

A multi-infinite tree always has at least one encounter-point. Moreover, it is easy to see, following \cite{burton_keane}, that the number of encounter points in some finite domain in the graph is never higher than the maximal number of disjoint one-way infinite paths which intersect the boundary of the domain (although we can decompose differently these paths, the maximal number of paths is finite and well defined).
\end{rmk}

%%-------------------------------------------------------------------------------------------------------
\subsection{Main Results}\label{sec:Main Results-inv subgraphs}
$~$\\
%In the subsection that handles translation invariant measure of bi-infinite paths in the plain we show that (under some ergodicity assumption), starting at any point on the path (in the support of the measure), any ray from this point, which is parallel to one of the axis, is a.s intersected infinitely many times by each of the two "half paths" which begin at the given point. Then we give some quadratic bound on the path-distance from the original point to the $n^{th}$ intersection point.
In the subsection that handles translation invariant measure of bi-infinite paths in the planar lattice we prove two main results. The first one, which concerns intersection of paths from the support of the measure with rays is the following theorem:
\begin{theorem}\label{thm:past intersects a ray infinitely many times}
Let $\mu$ be a measure of infinite paths in the $\mathbb{Z}^{2}$ lattice. Assume that $\mu$ is ergodic with respect to the group of $\mathbb{Z}^{2}$-translations and invariant under the group of rotations of the plane around the origin by integer multiples of $\pi/2$. Let $v$ be a vertex of the path, then the past of $v$ intersects the ray which is parallel to the $y$-axis, starts at $v$ and whose direction is upwards (intersects the lower half lattice only finitely many times) infinitely many times. Similar propositions hold for the other rays and the future as well.
\end{theorem}
The second theorem shows a quantitative result about the 'growth' of paths in the support of a translation invariant measure.
\begin{theorem}\label{thm:lim inf snail is at least quadratic}
Let $\mu$ be a translation invariant measure of bi-infinite paths in the $\mathbb{Z}^{2}$-lattice. Let $P$ be a path in the support of the measure, and $p$ be a lattice point in the path. then there is a positive $c = c(p,P)$ such that for infinitely many values of $n\in \mathbb{N}$, the length of the $n^{th}$-snail of $p$ is at least $cn^{2}$.
\end{theorem}
Next, we consider measures of forests in planar lattices, lattices in higher dimensions and in general amenable Cayley graphs. We show that the infinite connected components must be of some given shape (i.e. have no more than two infinite disjoint paths as subgraphs) and we analyze expected properties and sizes of subgraphs in these trees.
The main result of that part is the theorem below:
\begin{theorem}\label{thm:single-infinite at least boundary sized or nlogn}
Let $\mu$ be an ergodic , translation-invariant and rotation-invariant (as described above) measure of single-infinite trees on the lattice $\mathbb{Z}^{2}$. Then there exist a positive constant $\rho$, such that for any vertex $v$, and an increasing sequence of boxes, $\{B_n\}_{n=1}^{\infty}$, centered at $v$, the following holds:

Let $e_n$ be the number of edges of the tree which lay in the $n^{th}$ box and lay on a simple path in the tree which connects two boundary points of the box, then for all large enough $n$, $\mathbb{E}_\mu e_n > \rho nlog(n)$.
\end{theorem}
%%-------------------------------------------------------------------------------------------------------
\subsection{Translation Invariant Measures of Infinite $\mathbb{Z}^{2}$ Paths}\label{sec:infinite paths-inv subgraphs}
$~$\\
In this subsection we study translation-invariant measure supported on single two-sided infinite lattice paths. We consider the geometric properties of paths which are in the support of such measures. We prove some growth bounds on intersections of such paths with a fixed line, and we investigate intersections of such paths with other objects in the plain.

We begin by proving Theorem \ref{thm:past intersects a ray infinitely many times}. The proof is divided into several lemmas.

Under the assumptions of the theorem, denote by $H^{+}$ the ray through $v$ in the direction of the upper half $y$-axis. Denote by $H^{-}$ the second ray which is parallel to the $y$-axis and starts at $v$, but is a translation of the lower $y$-axis. We now state and prove some lemmas, each assumes the conditions of the theorem.
\begin{lemma}\label{lem:one ray -> all rays}
Assume that the origin is in the path, and that its past intersects the positive $y$-axis only finitely many times. Then for any point $(a,b)$ in the path, its past intersects the ray $\{(a,t)| t > b\}$ only finitely many times.
\end{lemma}
\begin{proof}
Assume the opposite. It is clear that if a point $(a,b)$ in the path has the property that its past intersects the half ray above it which is parallel to the $y$-axis only finitely many times, then this is true for any point $(a,c)$ in the path. It is therefore follows that finite intersection is a property of lines which are parallel to the $y$-axis. Due to translation invariance, and to the assumptions, there must be infinitely many rays which are a parallel translation (in the $x$ direction) of the positive $y$-axis such that the past of a point in the path intersects infinitely many times, and infinitely many other rays which the past does not intersects infinitely many times. In addition it also follows that these types alternate infinitely many times. In particular, there must be five such rays $a_1,a_2,a_3, b_1,b_2$ which are translations of the positive $y$-axis such that the rays which are indexed $a_i$ intersect the past of any point only finitely many times, and those indexed by $b_j$ are intersected infinitely many times, and that if we order these rays according to the $x$-coordinate we have
that $a_1$ is the leftmost, $b_1$ comes after it, then $a_2$,$b_2$, and $a_3$ (i.e. if ${a_i}^{x}$ is the $x$-coordinate of $a_i$, and similarly for $b_j$ then ${a_1}^{x}<{b_1}^{x}<{a_2}^{x}<{b_2}^{x}<{a_3}^{x}$).

Let $v\in\mathbb{Z}^{d}$ be a point in the path. There is some $h\in\mathbb{Z}$ such that the past of $v$ does not intersect the $a_i$ rays at height (i.e. $y$-coordinate) higher than $h$, and we assume in addition that $h$ is higher than the $y$-coordinates of lowest point of each $a_i$.
%Ron's remark
Define $T_j$ to be the times of intersections of the path with $b_j$ at points with $y$-coordinate bigger than $h$ (by 'time of intersection' we mean precisely minus the path distance from the intersection point to $v$), $j\in\{1,2\}$. Since the past of $v$ intersects each $b_j$ infinitely many times, and in particular intersects in points with $y$-coordinate bigger than $h$, it follows that
\\ $\forall t_1 \in T_1~\exists\{t_i\}_{i\geq 2}\subseteq T_1,~\exists\{s_i\}_{i\geq 1}\subseteq T_2~t_1>s_1>t_2>s_2...$

Now, since between any two times $t_i$ and $s_i$ as above, the path must intersect the segment parallel to the $x$-axis, at height $h$, that is between $a_1$ and $a_2$, but there are only finitely many points on this segment, whence the past must intersect itself. A contradiction.
%But the past of $v$ intersects each $b_j$ infinitely many times, and in particular infinitely many times at heights higher than $h$. But if we consider the set $T_1$ of intersection times with $b_1$ at points with $y$-coordinate higher than $h$, and $T_2$ which is similarly defined, we see that for any $t_1$ in $T_1$ there is a $t_2$ in $T_2$ which is smaller than $t_1$, and $t_3$ in $T_1$ smaller than $t_2$, and this can be continued infinitely many times. But between any to times $t_1$ and $t_2$ as above, the path must intersect the segment parallel to the $x$-axis, at height $h$, that is between $a_1$ and $a_2$ - but there are only finitely many points on this segment, hence the past must intersect itself - a contradiction. The result follows.
%???PICUTRE????????????????????
\begin{figure}[h]
\begin{center}
\includegraphics[width=0.7\textwidth]{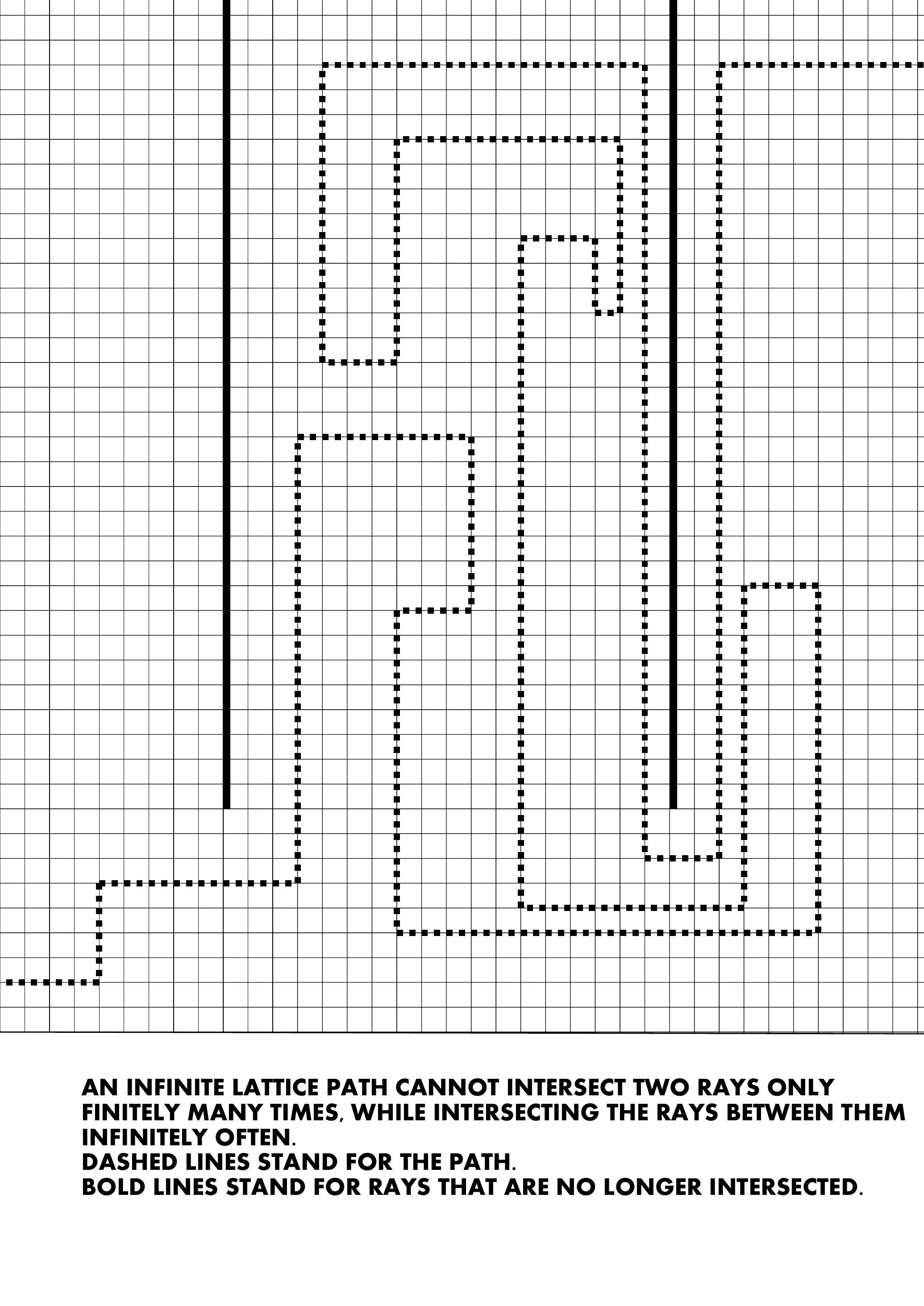}
\caption{Intersections of the past with rays - Lemma \ref{lem:one ray -> all rays}}
\end{center}
\end{figure}
\end{proof}
\begin{lemma}\label{lem:one ray -> not opposite ray}
If the past (of $v$) intersects $H^{+}$ only finitely many times, then it intersects $H^{-}$ infinitely many times.
\end{lemma}
\begin{proof}
Assume this is not the case, then after going backwards in time there is a point
$u\in P\bigcap\{H^{+}\bigcup H^{-}\}$ that all its past is either strictly to the left of $u$ or strictly to the right of $u$.  The set of points that all their past is strictly to their left will be called $PL$, similarly we define the set $PR$ for the right, and for the future $FL$, $FR$. Let $A = FL\cup FR \cup PL \cup PR$ their union. The event that the origin belongs to $A$ does not depend on the notion of 'direction', and hence is $\mu$-measurable (in Definition \ref{def:direction}, defining a direction required a choice, which may not be measurable). Let $\rho$ be the probability of that event. The ergodic theorem guarantees that for any large enough $N$, in the $N\times N$ square around the origin there are more than $\rho N^{2}/2$ points of $A$. On the other hand, the size of $PL$, and any of the other three sets is at most linear in $N$, since for any two points of $PL$ the $x$-coordinate must be different, as one must be strictly in the past of the other. Whence, $\rho$ must be $0$, and the lemma follows.
\end{proof}
\begin{lemma}\label{lem:one ray past -> not future}
If the past (of $v$) intersects $H^{+}$ only finitely many times, then its future intersects it infinitely many times.
\end{lemma}
\begin{proof}
Assume for contradiction this is not the case, then there is a point $u$ in the path, with the same $x$-coordinate as $v$, and with highest $y$-coordinate among all the points with that property. Let $\rho$ be the probability that the origin is the point in the path with the highest $y$ coordinate among those with $x$-coordinate $0$. An argument similar to the previous lemma shows that $\rho$ must be also $0$, but it can also be
shown in a dozen different ways.
\end{proof}
We now continue to the proof of the theorem.
\begin{proof}
Clearly all of the above lemmas can be proved for rays in the $x$-axis direction as well. In addition, due to the first lemma, if the past (future) intersects some ray which is parallel to the $y$-axis ($x$-axis) infinitely (only finitely) many times, it will intersect any other ray of the same direction infinitely (only finitely) many times.
There are several ways to finish the proof from here. One way is the following. Let $\mathcal{S}$ be the set of line segments which connect $(0,0)$ to the eight points $(i,j)\neq (0,0)$ where $-1\leq i,j \leq 1$. We define a function from the paths in the support of our measure to subsets of $\mathcal{S}$, which is translation invariant yet in general not rotation invariant (in multiples of $\pi/2$). The function is defined as follows:
\\ Let $\mathcal{X} = \mathcal{X}_{x,+}\bigcup\mathcal{X}_{x,-}\bigcup\mathcal{X}_{y,+}\bigcup\mathcal{X}_{y,-}$, where
\\ $\mathcal{X}_{x,+} = \{(x,0)\in\mathbb{Z}^{2}| x > 0\}$
\\ $\mathcal{X}_{x,-} = \{(x,0)\in\mathbb{Z}^{2}| x < 0\}$
\\ and the sets $\mathcal{X}_{y,+}, \mathcal{X}_{y,-}$ are defined in a similar way.
Given a path $P$, the past of any point in $P$ intersects at least two of the sets $\mathcal{X}_{x,\pm},\mathcal{X}_{y,\pm}$ infinitely many times, and in the case it intersects just two of them, then one of them is from $\mathcal{X}_{x,\pm}$ and the other is from $\mathcal{X}_{y,\pm}$.
% Let $P$ be the path. Let $\mathcal{X}$ be the set whose elements are the positive $x$ axis, the negative $x$ axis, the positive $y$ axis and the negative $y$ axis. Note that the past of any point in $P$ either intersects infinitely many times all the elements of $\mathcal{X}$, or three of them, or only two of them - but then one of the two is a $x$-half-axis and the other is a $y$-half axis. The same can be said about the future.
We now define the subset $S$ of $\mathcal{S}$, which is the image of $P$. In case the past of some point in $P$ intersects all the elements of $\mathcal{X}$ infinitely many times, it will not contribute an element of $\mathcal{S}$ to $S$. If it intersects three of $\mathcal{X}$'s elements infinitely many times, then it contribute to $S$ the segment which lays in the ray that it intersects only finitely many times (for example, if it were the negative $y$ axis, we would have put in $S$ the segment between the origin and $(0,-1)$). If it intersects two axis finitely many times many times we add the segment which lays on the (internal) bisector of these two rays (e.g. if the rays are the positive $x$ and $y$ axis, we add to $S$ the segment between the origin and $(1,1)$). We do the same procedure for the future.

%????PICTURE??????????????????
\begin{figure}[h]
\begin{center}
\includegraphics[width=0.7\textwidth]{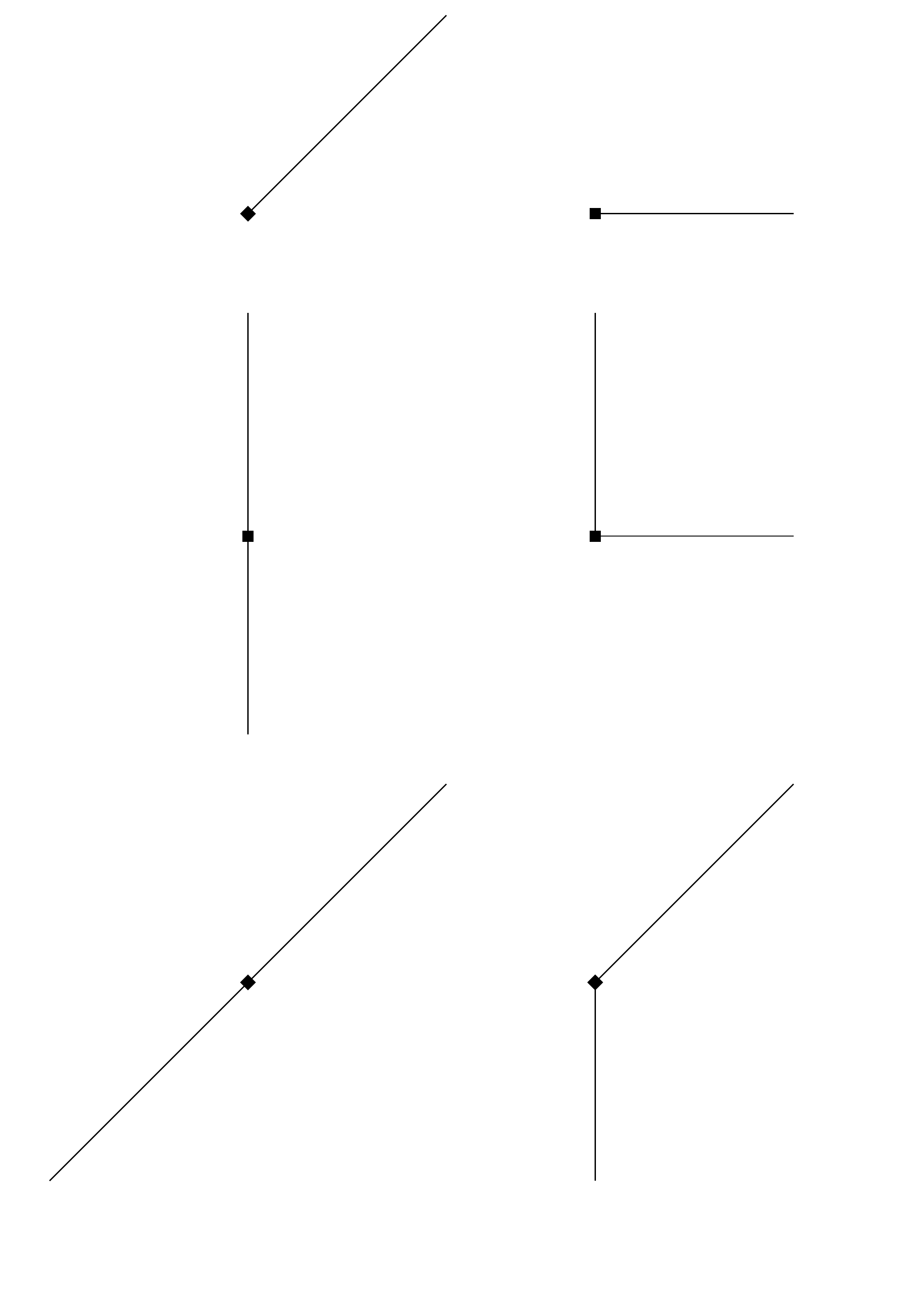}
\caption{Possible diagrams for the set S (Up to symmetries) - Theorem \ref{thm:past intersects a ray infinitely many times}}
\end{center}
\end{figure}

Note that the set $S$ does not depend on the direction of the path, i.e. on the choice of past and future. In addition, it is translation-invariant (due to the above arguments), but if we rotate the path in $\pi/2$ for example, we rotate (geometrically) the resulting set $S$ as well.
If either the past (of some point) or the future does not intersect at least one ray with positive probability, then the set $S$ has at least one element, but at most two. Due to ergodicity, the set $S$ must be constant a.s., on the other hand, no such set, which contains one or two elements of $\mathcal{S}$ returns to itself under a rotation of $\pi/2$, in contradiction with the translation invariance.
\end{proof}
\begin{rmk}
There is another proof to this claim, using a result of H. Kesten \cite{NOAM_KESTEN}.
\end{rmk}
Similar techniques lead to the following generalization
\begin{theorem}\label{thm:past intersects a rational ray infinitely many times}
Let $\mu$ be a measure of infinite paths in the $\mathbb{Z}^{2}$ lattice. Assume that $\mu$ is ergodic with respect to the group of $\mathbb{Z}^{2}$-translations and invariant under the group of rotations of the plane around the origin by integer multiples of $\pi/2$. Let $v$ be a vertex of the path, then the past of $v$ intersects any ray with a rational slope which starts at $v$ infinitely many times (intersections need not to be at lattice points this time). Similar proposition hold for the future.
\end{theorem}

Next we deal with measures of bi-infinite paths in the planar lattice, and try to achieve some growth bounds.
The next lemma, geometric in its spirit, will be our main tool for proving quantitative bounds on the behavior of translation invariant lattice paths.
\begin{lemma}\label{lem:cross lemma}
Let $P$ be a path as in the above definition. Let $p,q$ be two points of $P$ and assume that the $n^{th}$-cross of $p$ and the $m^{th}$-cross of $q$ intersect (perpendicularly) in a point. Then their snails intersect , and in particular the path distance between $p$ and $q$ is not more than the sum of the lengths of the corresponding snails.
\end{lemma}
Though it is easy to be convinced in the correctness of the Lemma, the proof is slightly tricky.
\begin{proof}
We show that the snails intersect, and the rest of the claim follows immediately from this fact. Denote by $O$ the intersection point of crosses. Assume without loss of generality. that $O$ is on the segment between $p$ and ${p_n}^{x_+}$, and that $O$ is on the segment between $q$ and ${q_n}^{y_-}$. Denote by $p'$ the point ${p_n}^{x_+}$, and by $q'$ the point ${q_m}^{y_-}$. From now on (in this proof) when we talk about the $\emph{partial}~\emph{snail}$ of $p$ we mean the minimal connected subpath that connects the points $p, {p_1}^{x_+},...,{p_n}^{x_+}$, and similarly for $q$ (but in the direction of the lower $y$-axis). We show that even these partial snails intersect.
If $O$ is on either partial snail, we are done, since in this case it is in the path, and it must be one of ${p_1}^{x_+},...,{p_n}^{x_+}$ and one of ${q_1}^{y_-},...,{q_m}^{y_-}$.

In addition, if the partial snail of $p$ intersects the segment between $q$ and $q'$, we are done as well, as again this is a path point which belongs to the two snails.

Thus, we may assume, towards contradiction, that the partial snail of $p$ does not intersect the segment between $q$ and $q'$ (and in particular not $O$), and the same assumption for the snail of $q$. Thus, if we consider the graphs of the union of the partial snail of $p$ with the segment between $p$ and $p'$, and the corresponding graph of $q$, their only intersection is $O$.

%!!!!!!!!!!!!!!PICTURE!!!!!!!!!!!!!!!!!
\begin{figure}[h]
\begin{center}
\includegraphics[width=0.7\textwidth]{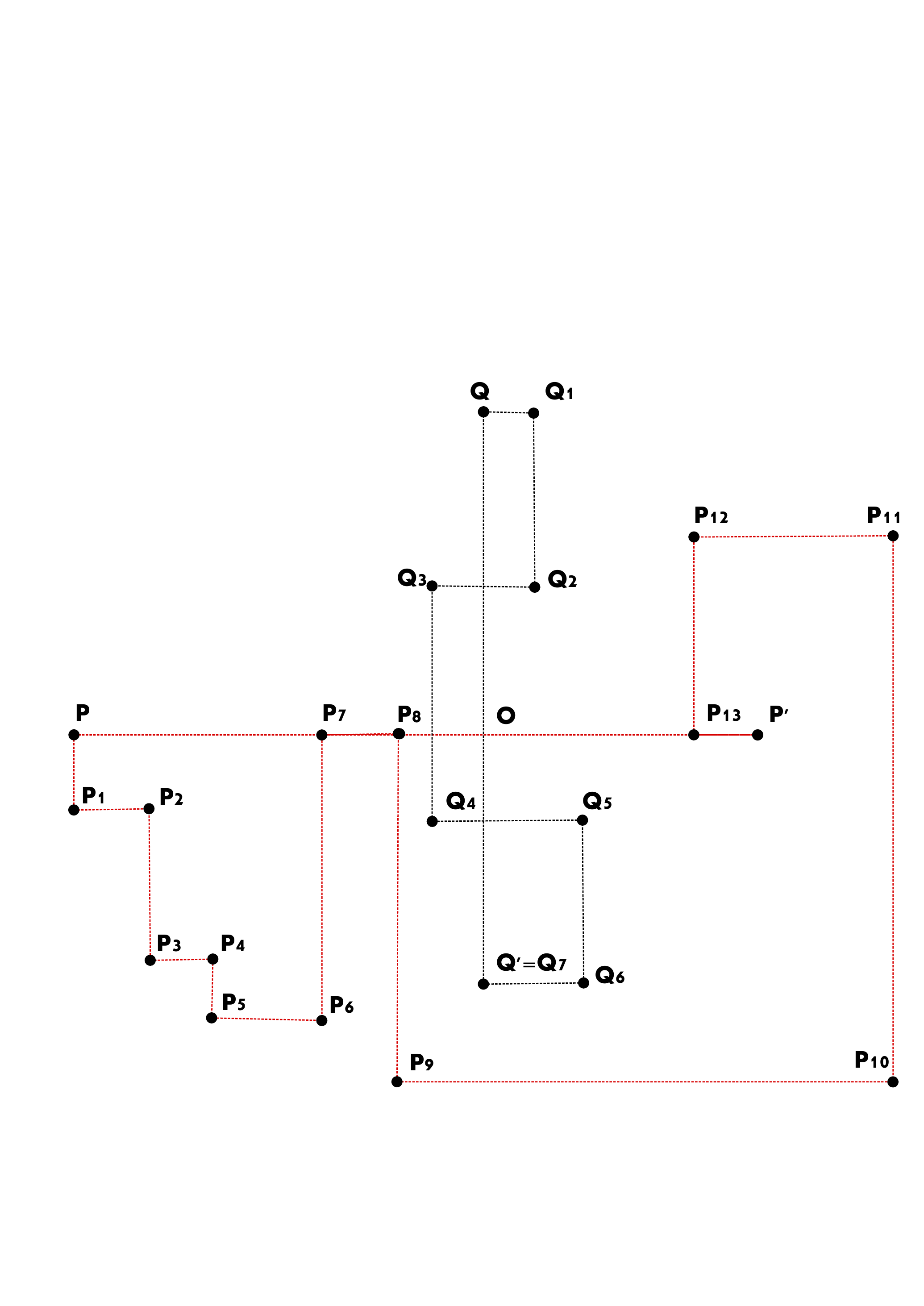}
\caption{Two intersecting snails - Lemma \ref{lem:cross lemma}}
\end{center}
\end{figure}
We can deform the graph (and by 'graph' we mean the picture of the graph, not a graph in the sense of graph theory yet) of $p$ slightly so that it remains a piecewise smooth graph (polygonal), such that the segment between $p$ and $p'$ remains in place, but the intersections in this graph are in a point, and there are not common segments to the deformed path and the $p$-$p'$ segment (e.g. if the path intersects the segment, lays on is for several edges and then cross to its other side - we leave one intersection point, and move slightly the rest of the curve. If after laying on the segment it does not cross, but returns to the previous side - we eliminate the intersection by moving that part of path slightly).
We do it for the other graph as well, and in a way that no new intersections are created or destroyed (between the two graphs). Note that $p, p',q,q'$ remain in place, and no deformation occurs in a neighborhood of $O$.
%!!!!!!!!!!!PICTURE!!!!!!!!!!!!!!!!!!!!!
\begin{figure}[h]
\begin{center}
\includegraphics[width=0.7\textwidth]{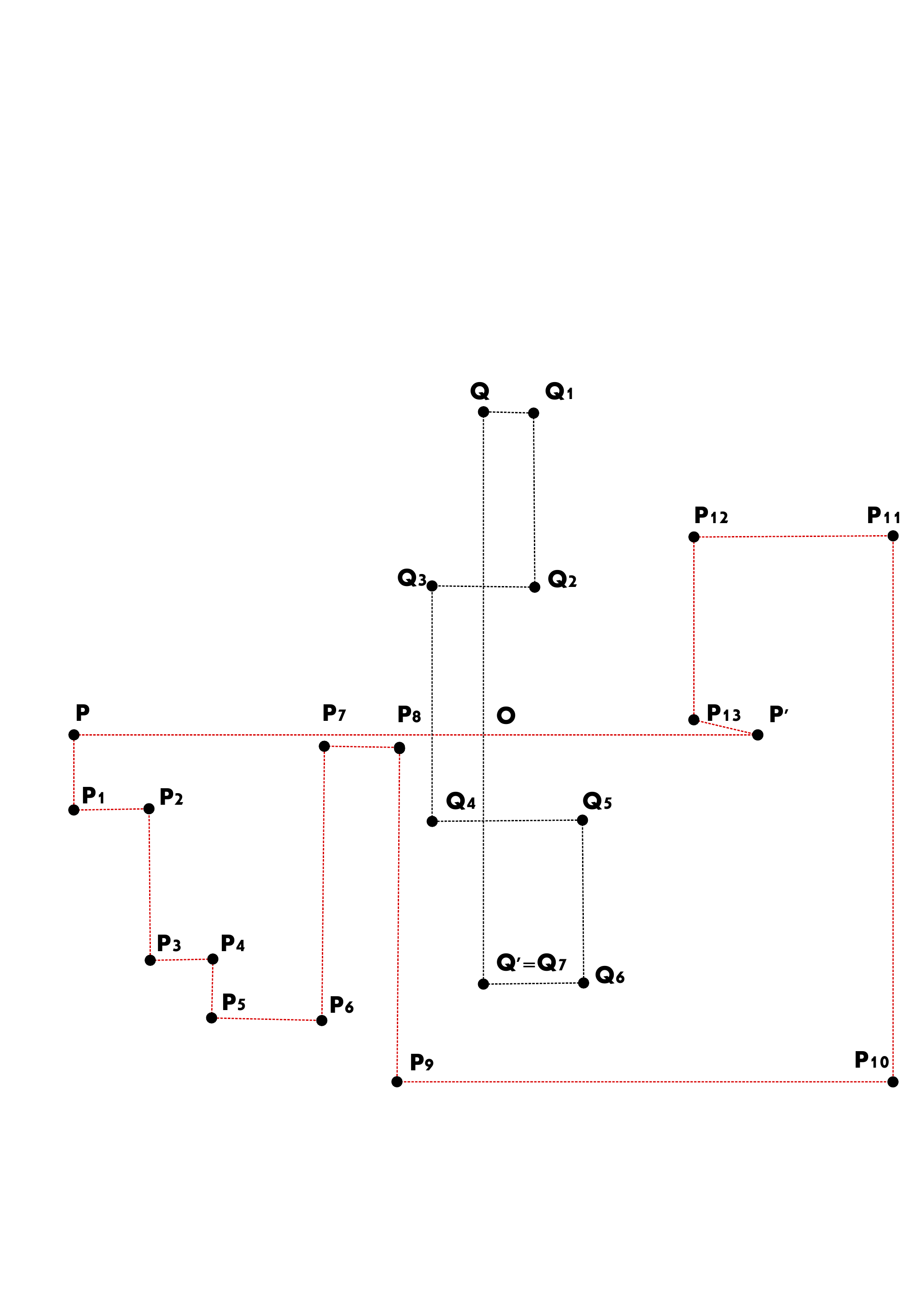}
\caption{Deformation of the intersection of figure 3 - Lemma \ref{lem:cross lemma}}
\end{center}
\end{figure}
Now we apply a graph theoretic argument. Consider the deformed graph of $p$, define self intersection points and points of intersection between the segment and the rest of the graph as vertices. Curves between them will stand for edges. This is a planar graph (with possibly multiple edges between vertices). It is easy to verify that every degree in the graph is either $2$ or $4$, and in particular - even. Hence it has an Euler-cycle. It is easy to see that such a cycle, in a planar graph can be decomposed as union of disjoint planar loops, such that none of which intersects itself, and mutual intersections are only in vertices. A similar argument can be applied on the deformed graph of $q$.

The total number of intersections of the two deformed graphs, is the sum over pairs of loops (one from each graph) of their intersections, as they have no common vertex. It is a well known fact that number of intersections of two transversal generic loops in the plain (e.g. polygonal closed paths, nowhere tangent to each other and whose intersection is made of a finite set of points, as in our case) is always an even number. Thus, their sum is even as well. But this is a contradiction to the fact that the only intersection point $O$ (and it is a generic intersection point). And the result follows.
\end{proof}
One may wonder about how the length of the path which connects $p$, a point on the path, and ${p_n}^{x_+}$ behaves. It is not difficult to verify that as a function of $n$, the expected (path) distance between an arbitrary vertex $p$ which is on the path and ${p_1}^{x_+}$, the closest intersection of the path with the positive $x$-ray from $p$ must be infinite, and hence the expected length of the path which connects $p$ and ${p_n}^{x_+}$ must be super linear. But can we say more?

A natural guess would be that if we consider the length of the path which connects $p$ and ${p_n}^{x_+}$, it will always be at least $cn^2$, for some positive $c$, for large enough $n\in\mathbb{N}$. A construction of B.Weiss shows that this is not the case (\cite{Weiss}).
Yet we prove that a weaker claim does hold, that infinitely many times, this inequality holds if we consider the $n^{th}$-snail. This is Theorem \ref{thm:lim inf snail is at least quadratic} that we state again below. It is followed by several conclusions, which show, for example, that if we add some more restrictions on the paths we can show that the path-distance between $p$ and ${p_n}^{x_+}$ is at least $cn^2$ infinitely many times.
\begin{theorem}\label{thm:lim inf snail is at least quadratic}
Let $\mu$ be a translation invariant measure of bi-infinite paths in the $\mathbb{Z}^{2}$-lattice. Let $P$ be a path in the support of the measure, and $p$ be a lattice point in the path. then there is a positive $c = c(p,P)$ such that for infinitely many values of $n\in \mathbb{N}$, the length of the $n^{th}$-snail of $p$ is at least $cn^{2}$.
\end{theorem}
\begin{proof}
Without loss of generality we assume that $\mu$ is ergodic, otherwise we apply ergodic decomposition.

For any positive $\varepsilon$, we define the set $A_{\varepsilon}$ as the set of points in $P$ such that for only finitely many values of $n\in\mathbb{N}$, the length of the $n^{th}$-snail is at least $\varepsilon n^{2}$. We define $A_0$ as
$\bigcap_{\varepsilon >0}A_\varepsilon$. The theorem would follow if we prove that $A_0$ is empty. Define $A_{\varepsilon , N}$ to be the set of points in $A_\varepsilon$ such that for any $n\geq N$, the length of their $n^{th}$-snail is not more than $2 \varepsilon n^{2}$. It is clear that $A_\varepsilon = \bigcup_{N\varepsilon\mathbb{N}}A_{\varepsilon,N} $, and that $A_{\varepsilon,N}$ is increasing in $N$. Hence, if we prove a uniform bound (in $N$) on the probability that the origin is in $A_{\varepsilon ,N}$, this bound will hold for $A_{\varepsilon}$ as well. If this bound tends to $0$ with $\varepsilon$, then the probability that the origin is in $A_0$ is $0$ as well. Thus, due to translation invariance, $A_0$ will be empty a.s. We now turn to prove the required bounds.
\begin{lemma}\label{lem:the density of A_eps}
Under the above definitions and conditions, the probability that the origin is in $A_{\varepsilon , N}$ (given it is on the path) is bounded by $8 \varepsilon$.
\end{lemma}
\begin{proof}
Denote by $\rho$ the probability that the origin is in $A_{\varepsilon , N}$. Due to the ergodic theorem, for large enough $M$ (that in particular we assume that is larger than $N$), in the $M\times M$ square around the origin there are at least $\rho M^{2}/2$ lattice points from $A_{\varepsilon , N}$. Since for any such two points, their $M$-crosses must intersect, it follows from Lemma \ref{lem:cross lemma}, that the path distance between these two points is no more than the sum of the lengths of their snails. As these two points are in $A_{\varepsilon , N}$, the sum of lengths of their snails is bounded by $4 \varepsilon M^{2}$.

Thus, the diameter of the connected subgraph of the path that contains all the points of $A_{\varepsilon , N}$ which are in the $M \times M$ square is at most $4 \varepsilon M^{2}$. But any connected subgraph of a path is a path, and its diameter is its length. Thus, we have found a path of length at most  $4 \varepsilon M^{2}$, which contains at least $\rho M^{2}/2$ vertices. Hence, $8 \varepsilon \geq \rho$
\end{proof}
And thus, the theorem is proved.
\end{proof}
The above theorem and its way of proof lead to some conclusions.
\begin{conclusion}
Under the above conditions, if we also assume that the measure is ergodic with respect to translation by every group element (not only ergodic with respect to the whole group's action), then there is a positive $c = c(\mu)$ such that for $\mu$-almost every path $P$, any point $p$ on the path belong to $A_c$.
\end{conclusion}
\begin{proof}
Let $B_\varepsilon$ be the complement of $A_\varepsilon$.
We define ${B_\varepsilon}^{x}$ as the set of points $p$ such that for infinitely many values of $n$ the minimal connected subpath that contains all the points $p,{p_1}^{x_+},..., {p_n}^{x_+},{p_1}^{x_-},...,{p_n}^{x_-}$ is of length at least $\varepsilon n^{2}$. Similarly one defines ${B_\varepsilon}^{y}$.
We have the following simple properties:
\\ 1. If $p$ is in $B_\varepsilon$ then either $p$ is in ${B_{\varepsilon /2}}^{x}$ or in ${B_{\varepsilon /2}}^{y}$.
\\ 2. If $p$ is in ${B_\varepsilon}^{x}$ or ${B_\varepsilon}^{y}$ then $p$ is in ${B_\varepsilon}$.
\\ 3. If $p$ is in ${B_\varepsilon}^{x}$ then any other (path) point with the same $y$ coordinate is also in ${B_\varepsilon}^{x}$. If $p$ is in ${B_\varepsilon}^{y}$ then any other point with the same $x$ coordinate belongs to ${B_\varepsilon}^{y}$.

Assume that the origin is in the path. Due to the above theorem, there is some $\varepsilon > 0$, for which the origin is in $B_\varepsilon$. Due to property 1, and without loss of generality the origin is in ${B_{\varepsilon /2}}^{x}$. Consider the event that the origin is in ${B_{\varepsilon /2}}^{x}$, it is invariant under translations of direction $x$. Thus, by ergodicity in each direction separately this is a $0-1$ event. But it has a positive probability (for some $\varepsilon$), and therefore it occurs a.s. Due to invariance to the other direction we see that a.s. any path point is in ${B_{\varepsilon /2}}^{x}$. And hence, due to property 2, a.s. any path point is in ${B_{\varepsilon /2}}$.
\end{proof}
\begin{conclusion}
Under the assumption of the previous conclusion, if in addition the distribution $\mu$ is invariant under rotations around the origin by angles which are integer multiples of $\pi /2$ then there exist a unique $\varepsilon = \varepsilon (\mu)$ such that for every point in the path, for any positive $\delta < \varepsilon$, for each $\alpha\in\{ x_+, x_-, y_+, y_- \}$ there are infinitely many values of $n$, such that the length of the minimal connected subpath which contains $p, {p_1}^{\alpha},...{p_n}^{\alpha}$ is at least $\delta n^{2}$, and such that no path point that property for any $\delta > \varepsilon$.
\end{conclusion}
The proof is similar to the above proof, and the rotation invariance guarantees that there is a uniform positive $\varepsilon$ to each of the four directions.

\begin{conclusion}
Assume $\mu$ is rotation-invariant, as in the above conclusion, and ergodic with respect to the group's action (and not necessarily to each direction separately), then almost surely, for any point $p$ in the path $P$ there exist positive ${\varepsilon}_x , {\varepsilon}_y$, which depend on the path and the point such that, in the notation of the first conclusion, $p$ belongs to
${B_{{\varepsilon}_x}}^{x} \bigcap {B_{{\varepsilon}_y}}^{y}$.
\end{conclusion}
\begin{proof}
$\mathbf{Sketch.}$ The event that there exist a point with no positive ${\varepsilon}_x$ for example, is translation invariant. Hence it must be a $0-1$ event. If it were an event of probability $1$, then the same would hold for points with no positive ${\varepsilon}_x$. As in the first conclusion, if $p$ has no positive $\varepsilon_x$ with the given property, then the same holds for any other point with the same $y$ coordinate as $p$. Thus, it is a property of lines in the lattice which are parallel to the $x$-axis. Thus, with probability $1$, there should be such a line. Similarly, there should be a line with a similar property which is parallel to the $x$ axis. Denote by $q$ their intersection. From property 1 in the first conclusion it follows that $q$ is not in any $B_\varepsilon$. But this contradicts Theorem \ref{thm:lim inf snail is at least quadratic}.
\end{proof}

%%-------------------------------------------------------------------------------------------------------
\subsection{Translation Invariant Measures of Infinite Lattice Trees}\label{sec:infinite trees-inv subgraphs}
$~$\\
This subsection investigates measures of trees and forests which are subgraph of a given graph, usually the graph $\mathbb{Z}^{d}$. Most of the results in this subsection are rather simple, but allow us to imagine the geometric picture trees which arise in translation invariant contexts. The last theorem of the subsection is more interesting and also slightly more sophisticated.
We consider measures, invariant under the graph's group of symmetries (or a subgroup of it), of infinite trees or forests whose connected components are infinite trees. Some of the results hold for a general unimodular transitive graph, some are proved for amenable Cayley graph and some hold for Euclidean lattices (and in particular the planar square lattice).

\begin{claim}
Let $G$ be a connected Cayley graph of an amenable group. Let $\mu$ be a measure on forests which are unions of infinite trees in $G$. Assume that $\mu$ is invariant under the groups's action. Then in almost every forest which is in the support of the measure, any tree component is either single-infinite or bi-infinite.
\end{claim}
\begin{proof}
By ergodic decomposition we may assume that $\mu$ is ergodic.
Assume there are some multi infinite components, in that case there are encounter points as well.
Let $F$ be any subgraph of the graph $G$. By the last part of Remark \ref{rmk:rmk forests}, $\mu$-a.s the number of encounter points inside $F$ cannot be larger than the number of boundary edges of $F$ (this is actually always true, not just $\mu-$a.s). Thus, the $\mu$-expectation of the number of encounter points inside $F$ is not more than the number of boundary edges. On the other hand, if we assume that the probability for having encounter points is positive, then
%with positive probability there are encounter points,
there exists a positive number $\rho$ such that an arbitrary vertex $v$ is an encounter point with probability $\rho$. The linearity of expectation yields that the expected number of encounter points in $F$ is $\rho |F|$. Thus $| \partial F| \geq \rho |F|$, where $\partial F$ is the set of boundary edges of $F$.
Let $\{ F_n \}$ be a F$\phi$lner sequence. The above argument shows that $| \partial F_n| / |F_n| \geq \rho$. But this contradict the definition of F$\phi$lner sequences, where the LHS of the inequality must tend to $0$.
\end{proof}
We turn examine the shape of single-infinite trees (or forests made of infinite trees), that are in the support of an invariant measure (under the group actions). Our assumption on the full graph are that it is a (unimodular) Cayley graph. Our main tool will be the Mass-Transport principle. 

\begin{claim}\label{clm:trees MT}
Let $G$ be a unimodular transitive graph (a Cayley graph of a discrete group). Let $\mu$ be a measure, invariant under the automorphism group of the graph (the group's actions), of forests made of single-infinite trees. Let $f:\mathbb{N}\rightarrow \mathbb{R}$ be any nonnegative valued function. Then
\begin{equation}
\forall x\in V~ \sum_{k\in\mathbb{N}}f(k) = \mathbb{E}_\mu[\sum_{y\in R(x)\backslash{x}}f(\delta(y,x))]
\end{equation}
In particular, one side is infinite iff the second side is.
\end{claim}
\begin{proof}
It is a direct conclusion from Equation \ref{eq:MT} if we define the mass function
\\ $m(\omega;x,y)=f(\delta_\omega(x,y))\times\mathbf{1}_{x\in {S(y)}}$
\\ where $\delta_\omega(\cdot,\cdot)$ is the tree distance between $x,y$ in configuration $\omega$, and it is defined to be $0$ if $x,y$ are not in the same tree (or not in a tree), and $\mathbf{1}_{x\in {S(y)}}$ is the characteristic function of the event that $x$ is in the stem of $y$ (in configuration $\omega$).
\end{proof}
We give some conclusions from the above claim. In all the conclusions we assume the assumptions of Claim \ref{clm:trees MT}.

\begin{conclusion}\label{con:infinite functions of root}
Let $f:\mathbb{N}\rightarrow \mathbb{R}$ be such that $\sum_{k\in\mathbb{N}}f(k)=\infty$, then
\\ $\mathbb{E}_\mu [\sum_{y\in R(x)\backslash{x}}f(\delta(y,x))]=\infty$
\end{conclusion}
%\begin{exam}
%$\mathbb{E}_\mu\sum_{y\in R(x)\backslash{x}}f(\delta(y,x))=\infty$
%for $f(n)=1/nlog(n+1)$.
%On The other hand this expectation is finite for $f(n)=n^{-2}$ or $f(n)=1/n(logn)^{2}$
%\end{exam}

\begin{conclusion}\label{con:infinite expected root}
%The $\mu$ expected size of the roots of some vertex is infinite.
For an arbitrary vertex $v$ the $\mu$-expected size of $R(v)$ is infinite.
\end{conclusion}
\begin{proof}
This follows immediately from the above conclusion, just by taking the constant function $f\equiv 1$, but we give an equivalent proof, in order to improve the intuition for the scenario.

If each vertex of the forest sends one unit of mass to each vertex in its stem, since any vertex in a tree sends an infinite amount, and any vertex has a positive probability to be in the tree, it follows that the expected total mass that each vertex in the forest sends is infinite.
Thus, the MT principle tells us that any vertex receives on average an infinite amount of mass. But any vertex not in the forest receive nothing, and a vertex $v$ in the forest receives exactly the size of $R(v)\backslash{v}$.
\end{proof}

\begin{conclusion}\label{con:one expected root element of dist n}
Let $v$ be a vertex in the graph. The $\mu$-expected value of the number of vertices in $R(v)$ (conditioned on $v$ being in the forest), with tree-distance $n$, is $1$.
\end{conclusion}
\begin{proof}
This follows immediately from Claim \ref{clm:trees MT}. If each vertex of the forest sends one unit of mass to the (unique) vertex in its stem that is in distance $n$ from it. Any vertex of the forest sends exactly one unit of mass. The MT principle tells us that a vertex, conditioned on being in the forest, receives an expected number of one unit of mass. This implies the claim.
\end{proof}

\begin{conclusion}\label{con:lower bounds for large root in amenable graphs}
Let $\alpha^{v}(n)$ be the number of vertices in a ball of radius $n$ in the graph $G$ around a given vertex $v$, i.e. the number of vertices in the graph, whose shortest path to $v$ is of length at most $n$. Let $\beta^{v}(n)$ be the size of the sphere of radius $n$ around $v$ in the graph, i.e. the set of vertices in the graph, whose shortest path to $v$ is of length exactly $n$. Then the $\mu$-probability that a given vertex $v$, conditioned on being in the forest, has at least one vertex in $R(v)$ of tree-distance $n$ is not less than $1/\alpha^{v}(n)$. The $\mu$-probability that a given vertex $v$, conditioned on being in the forest, has at least one vertex in $R(v)$ of graph-distance $n$ is not less than $1/\beta^{v}(n)$.
\end{conclusion}
\begin{rmk}
Note that in $R(v)$ there is a vertex whose graph (tree) distance from $v$ is at least $n$, is exactly the probability there is a vertex in $R(v)$ whose graph (tree) distance from $v$ is exactly $n$, since in the path between $v$ and $u$ which is of length more than $n$, there is at least one vertex whose distance from $v$ is $n$.
\end{rmk}
\begin{proof}
First note that since we always assume that $G$ is transitive $\alpha^{v}(n)$ does not depend on the vertex $v$, and hence we may denote it simply by $\alpha(n)$. Similarly to $\beta^{v}(n)$.
Second observe that for any tree $T$ which is a subgraph of a graph $G$, the tree-distance of two tree vertices is at least the graph distance, i.e. $\delta_T(x,y)\geq d(x,y)$. Take a vertex $v$, and consider the ball of radius $n$ around it (in $G$). denote by $\alpha$ the number of the vertices in $B$. Any vertex in $R(v)$, which has a tree distance exactly $n$ from $v$, its graph distance from $v$ is at most $n$, and hence it belongs to $B$. Thus, the number of elements of $R(v)$ whose tree distance from $v$ is exactly $n$ is (nonnegative, and) bounded from above by $\alpha$ and by Conclusion \ref{con:one expected root element of dist n} has expectation $1$. Now the Markov inequality guarantees that the probability there are vertices in $R(v)$ of tree-distance exactly $n$ is at least $1/\alpha$.

We would like to achieve the corresponding result for graph-distances. For this first consider the following mass function: Each vertex in the tree sends one mass unit to each vertex in its stem whose graph distance from it is exactly $n$. It is clear that each vertex in the tree sends at least one mass unit, and hence, in average, according to Claim \ref{clm:trees MT} receives at least $1$ unit of mass. This time we consider $\beta$, which is the number of vertices of the $n-$sphere centered in some given vertex $v$. Again $\beta$ bounds the number of vertices in $R(v)$ of graph distance $n$ from $v$, and again Markov's inequality shows that the probability that there is at least one vertex in $R(v)$ of graph-distance $n$ from $v$ is at least $1/\beta$.
\end{proof}

\begin{conclusion}\label{con:lower bounds for large root in lattices}
If the graph $G$ is the $\mathbb{Z}^d$-lattice, then for a vertex $v$ in the forest the probability that there is some element of $R(v)$ with graph distance at least $n$ is $\Omega(n^{1-d})$. In addition, the probability that there is an element of $R(v)$ with tree-distance at least $n$ is $\Omega(n^{-d})$
\end{conclusion}
\begin{proof}
The ball (box) of radius $n$ in the $\mathbb{Z}^d$-lattice has $\Theta(n^{d})$ vertices and its boundary has $\Theta(n^{d-1})$ vertices. Using Conclusion \ref{con:lower bounds for large root in amenable graphs} we get the result immediately.
\end{proof}

\begin{rmk}
Most of the above can be established in a similar way to a bi-infinite tree. In a bi-infinite tree $T$, there is a unique path between any vertex and $P(T)$. Thus, even though we do not have a canonic "forward" direction, as in the single-infinite case, we do have two "forward" directions, that may have a finite path in common. Indeed, we can move "forward" in the path which connects it to $P(T)$, after we reach $P(T)$ we have two legal directions towards infinity. Thus, one can construct mass functions as we did for this case as well, such that mass is again sent only forward.
\end{rmk}

Our next results handle densities of bi-infinite and single-infinite trees which are in the support of some invariant measure. As before, due to ergodic decomposition we may assume the measure is ergodic. The main question we are interested in is the following - if we consider a large ball in the graph, and count edges which lay on simple paths in the tree which connect two of the ball's boundary points, how many such edges are there? We show that in the double-infinite case the number is proportional to the volume of the ball. In the single infinite it is at least a fixed portion of the boundary (we consider only lattices for simplicity). The surprising result is that for the two dimensional lattice we show a superlinear bound. A generalization of this fact plays an important role in our study of spin systems. All these results can be easily generalized to forest of infinite trees as well.

\begin{claim}\label{clm:bi-infinite have densities}
Let $\mu$ be an ergodic measure of bi-infinite trees on a connected Cayley graph of
%an amenable
a group, invariant under the group actions. Then there exist a positive constant $\rho$, such that for any vertex $v$, and any finite subgraph $F$ the following holds:
%an increasing sequence of balls, $\{B_n\}_{n=1}^{\infty}$, the following holds:
%any F$\phi$lner sequence $\{ F_n \}$, the following holds:

Let $S$ be the $\mu$-expected number of edges of the tree which belong to $F$ and lay on a simple path in the tree which connects two boundary points of $F$, and $E$ is the number of all edges of $F$, then $S \geq \rho E$.
\end{claim}
\begin{proof}
$\mathbf{Sketch.~}$The idea is that any edge in $P(T)$, the path of the tree, has the above property. And any edge has a positive probability to be in $P(T)$. Thus the expected number of edges that both of whose vertices are in $F$ and belong to $P(T)$ is a fixed, positive portion of the size of $E$.
\end{proof}
\begin{rmk}
If we assume that the above Cayley graph is also amenable, an ergodic theorem due to Lindenstrauss (\cite{pointwise_erg}) guarantees that by replacing the single set $F$ with a F$\phi$lner sequence, we get that the above conclusion holds not only in expected value but also a.s. for some subsequence, i.e. there is an infinite subsequence of the given F$\phi$lner sequence such that for almost every set in the subsequence the number of edges which lay on simple paths that connect two boundary edges is at least a $\rho-$portion of all the edges in that set.
\end{rmk}
Our last result is the main theorem of this subsection. A variant of this theorem will appear later as one of the key stones in the proof that in the ground-state of a planar Ising spin-glass system, there is no infinite component made of edges which are all unsatisfied. Note that we have stated a partial version of this theorem in Subsection \ref{sec:Main Results-inv subgraphs}, we now state and prove a slightly more general theorem.
For simplicity the next result is formulated for the lattice $\mathbb{Z}^{d}$, and for measures which are also invariant under rotations (of multiples of $\pi/2$ in any integer axis). Yet, it can be be extended to larger families of graphs, and the rotations-invariance is, in fact, not necessary.
\begin{theorem}\label{thm:single-infinite at least boundary sized or nlogn}
Let $\mu$ be an ergodic , translation-invariant and rotation-invariant (as described above) measure of single-infinite trees on the lattice $\mathbb{Z}^{d}$. Then there exist a positive constant $\rho$, such that for any vertex $v$, and an increasing sequence of boxes, $\{B_n\}_{n=1}^{\infty}$, centered at $v$, the following holds:

Let $e_n$ be the number of edges of the tree which lay in the $n^{th}$ box and lay on a simple path in the tree which connects two boundary points of the box, then for all large enough $n$, $\mathbb{E}_\mu e_n > \rho n$.

Moreover, for $d=2$, we have $\mathbb{E}_\mu e_n > \rho nlog(n)$.
\end{theorem}
We prove the result for the case $d=2$ and the case $d > 2$ together. Yet there is a simpler proof for the latter case (which is also less interesting).
\begin{proof}
\begin{figure}[h]
\begin{center}
\includegraphics[width=0.7\textwidth]{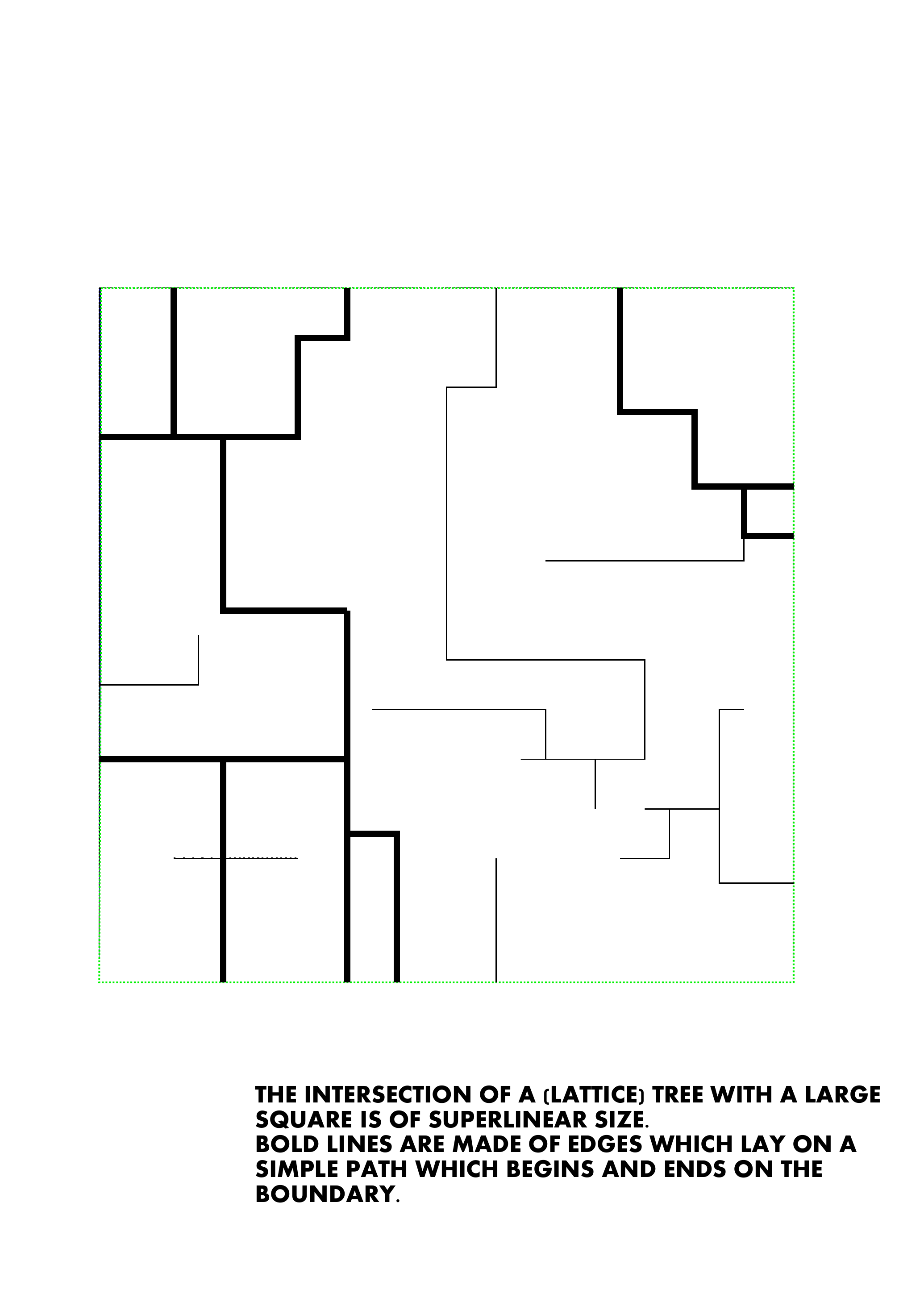}
\caption{The intersection of an infinite lattice tree with a large square - Theorem \ref{thm:single-infinite at least boundary sized or nlogn}}
\end{center}
\end{figure}
First we observe that similarly to Conclusion \ref{con:lower bounds for large root in lattices} we can show that with probability of $\Omega(n^{1-d})$ for a vertex $v$ in the single-infinite tree there is a vertex $u\in R(v)$ with $\| u-v \|_\infty = n$, this is done in the standard way, using the mass function $m(u,v)$ which is always $0$ unless $v\in S(u)$ and $\| u-v \|_\infty = n$, and then it is $1$. Thus, there is a positive constant $c_1$, such that for a given vertex $v$ and $n\in\mathbb{N}$, the probability there is another vertex $u\in R(v)$ such that $\| u-v \|_\infty = n$ is at least $c_1 n^{1-d}$.
 for any $n$ with probability at least $c_1 n^{1-d}$ for a given vertex $v$ there is another vertex $u\in R(v)$ such that $\| u-v \|_\infty = n$. In particular, due to the rotation-invariance, if we denote by $c := c_1/2d$, then for any $n$ with probability at least $c_1 n^{1-d}$ for a given vertex $v$ there is another vertex $u\in R(v)$ such that $\| u-v \|_\infty = n$, and that $u^{1}-v^{1}=n$, where $x^{i}$ is the $i^{th}$ coordinate of $x$ (and of course the same probability bound holds for any $u^{i}-v^{i}=\pm n$).
There is a constant $b>0$ such that for any $n$ the number of vertices in the boundary of a box of $l_\infty$ norm $n$ is at least $bn^{d-1}$.

Every vertex $v$ from the tree, which lays in a given box, is a part of a simple path (contained in the tree) which connects two points on its boundary iff there is another vertex from $R(v)$ which is in that boundary. The reason is that there is always a "half-path" from $v$ to the boundary - $S(v)$ (the part from $v$ until first intersection with the boundary of the box).

In addition, the number of tree edges which lay on simple paths which connect to boundary edges is clearly proportional to the number of vertices with the same property (as the degree in the graph is bounded by $2d$). Thus, it is enough to give a lower estimate for the expected number of vertices with that property. Denote by $p$ the probability a given vertex lays in $T$ (of course, $p$ does not depend on the vertex, it is an a.s. constant which depends only on $\mu$ and is the same for any point and any configuration in the support of $\mu$). Then the probability a given vertex lays in the tree, and in its roots set there is a vertex at $l_\infty$ distance $t$ from it, in a given direction is at least
$pc n^{1-d}$.

Consider a fixed vertex $v$, denote by $S_t$ the set of points with $l_\infty$ distance $t$ from $v$, $B_t$ be the box of vertices whose $l_\infty$ distance from $v$ is at most $t$. Combining all of the above gives that the expected number of tree vertices in a box, which lay on simple paths (simple paths which are contained in $T$) and connect two of $B_n$'s boundary points is at least
\begin{equation}
pbc\sum_{1\leq k \leq n-1}(n-k)^{d-1}/{k^{d-1}} > (n/2)^{d-1}\sum_{1\leq k \leq \lfloor n/2 \rfloor}k^{1-d}
\end{equation}
For $d=2$ this gives an $\Omega(nlog(n))$ bound. For higher $d$ is gives the expected linear bound. As required.
\end{proof}

%%-------------------------------------------------------------------------------------------------------
%\subsection{Some Generalizations to Transitive Amenable Graphs}\label{general graphs-inv subgraphs}
%$~$\\
\newpage
%%-------------------------------------------------------------------------------------------------------
\section{Spin Glass Models}\label{sec:spinglasses}
$~$\\

%%-------------------------------------------------------------------------------------------------------

\subsection{Background}\label{sec:Background-spinglass}
$~$\\
The understanding of spin glass models in large or infinite graphs has been the subject of many studies in physics, mathematics and neuroscience. Questions regarding the multiplicity of ground states in finite dimensional short-ranged systems, such as the Edwards-Anderson (EA) Ising spin glass, and in particular the 2D case were the subjects of several researches and simulations (e.g. \cite{EA}, \cite{Newman_Stein_plain}, \cite{Newman00natureof}, \cite{Newman_Stein_halfplain},\cite{PhysRevLett.58.57},\cite{0022-3719-17-18-010},\cite{PhysRevLett.56.1601},\cite{PhysRevLett.83.5126},\cite{PhysRevB.46.973}, \cite{HARTMANN}). Even less is known about the geometry of the ground states.

%%-------------------------------------------------------------------------------------------------------
\subsection{Preliminaries}\label{sec:Prelim-spinglass}
$~$\\
This part of the thesis concerns the EA Ising spin glass with the Hamiltonian
\begin{equation}
\mathcal{H}_\mathcal{J}(\sigma) = -\sum_{<x,y>}J_{xy}\sigma_x\sigma_y
\end{equation}
Where $\mathcal{J}$ denotes a specific realization of the couplings $J_{xy}$ and $\sigma$ is a spin configuration, i.e. for any vertex $x$ in the graph, $\sigma_x = \pm 1$. The sum is taken over nearest neighbor pairs $<x,y>$ in a given graph G. We confine ourselves to the case where $J_{xy}$ are independently chosen from a mean zero Gaussian or any other symmetric, continuous distribution supported by the real line. The overall disorder measure is denoted by $\nu(\mathcal{J})$.
A $\emph{Configuration}$ for some realization $\mathcal{J}$ is a choice of spin to each site. Although the Hamiltonian might not be defined, due to divergence, we can still compare the "Energy difference" (defined as the Hamiltonian's difference) between two configurations which differ by a finite amount of spins. A $\emph{Ground}~\emph{State}$ is a configuration whose energy cannot be lowered by flipping any finite subset of spins. That is, all ground state spin configurations must satisfy the constraint
\begin{equation}
-\sum_{<x,y>\in\mathcal{{\partial C}}}J_{xy}\sigma_x\sigma_y \leq 0
\end{equation}
for any closed finite subset of vertices $\mathcal{C}$. One should note that if we restrict our attention to a planar graph, ${\partial C}$ is actually a union of loops in the dual lattice. One should note that if $\sigma$ is a ground state, so is $-\sigma$, the configuration which is the result of flipping all the spins in $\sigma$. Thus, it makes more sense to talk about $\emph{ground state pair}$ or $\emph{GSP}$. The amount $J_{xy}\sigma_x\sigma_y$ in some configuration is sometimes called the $\emph{value}$ of the edge $\{x,y\}$ (in this configuration).

Let $\mu_\mathcal{J}$ be a general conditional (on $\mathcal{J}$) distribution which is translation-invariant and is supported on GSPs (or more precisely - the ground state representatives of them) for $\mathcal{J}$.

\begin{defn}
Let $G$ be a graph, $J$ a realization of couplings. An edge $e= \{ x, y \} $ is said to be $\emph{fixed}$, if either $|J_{xy}|>\sum_{z\sim x}|J_{xz}|$ or $|J_{xy}|>\sum_{z\sim y}|J_{yz}|$.
A subgraph $H$ is said to be $\emph{fixed}$ if there is a tree $T\subseteq G$ such that all the edges in $T$ are fixed and all the vertices of $H$ are vertices of $T$.

Under some spin configuration the product of an edge's interaction with its two vertices' spins is called its $\emph{value}$.

We say that a bond $e=\{x,y\}$ is $\emph{unsatisfied}$ (in some spin configuration $\sigma$) if $J_{xy}\sigma_x\sigma_y$ is negative. In this case we also say that the dual bond $e^{*}$ is unsatisfied.
\end{defn}
\begin{rmk}
In every GSP, the value of a fixed edge is positive and hence either in every GSP its vertices have the same spin, or in every GSP its vertices have opposite spins. In other wards, the products of its vertices' spins is the same in any GSP.
In a similar manner, all the edges of a fixed subgraph have the same value in each GSP and thus, the spins of $H$'s vertices equal each other.
%are the same in any ground state up to a global flip (as $T$ is connected).
\end{rmk}
\begin{defn}
Let $G=(V,E)$ be a graph, $S\subseteq V$ a subset of vertices. We say that $S$ bounds a subgraph $H$ if $H$ is a connected component of $G\mid_{V\setminus {S}}$.
\end{defn}
\begin{defn}
Let $G=(V,E)$ be a graph, the $\emph{boundary}$ of a subgraph is the set of edges which connect it to its complement in $G$.
\end{defn}

We continue by defining several families of graphs that we explore.
\begin{defn}
Let $H = (V_H,E_H)$ be a graph. We say that a graph $G=(V_G,E_G)$ is a $\emph{H-type Graph}$ if $V=\bigcup_{h\in V_H}{\{h\}}\times{V_h}$ and if there is an edge between $(h,v)$ and $(h',v')$ then either $h=h'$ or $h\sim_H h'$, where by $\sim_H$ we mean that $h,h'$ are neighbors in $H$.
 We denote by $G_h$ the subgraph of $G$ over the vertices with H-coordinate $h$. We call it the $\emph{the}~$ $h^{th}~$ $\emph{slice}$. For a vertex of the form $(h,v)$ we shall say that $h$ is its $\emph{level}$. The $\emph{width}$ of the $h^{th}$-slice is the size of $V_h$. The $\emph{width}$ of $G$ is defined as $sup_{h\in H} |V_h|$.

A $\mathbb{Z}$-type graph is called a $\emph{deformed cylinder}$.
\end{defn}
An important special case is the product of graphs, defined as follows
\begin{defn}
Let $H=(V,E),G=(U,F)$ be graphs. We define their $\emph{product}, G\times H$ as the graph whose vertex set is $V\times U$ and there is an edge between $(h,g), (h',g')$ iff exactly one of the following occurs - $h\sim_H h'$ and $g = g'$ or $g\sim_G g'$ and $h = h'$.

$\mathbb{Z}\times G$ will be called a $\emph{G-cylinder}$ or simply a $\emph{cylinder}$.
$K_n$ will denote the complete graph on $n$ vertices. $C_n$ will stand for the cycle of length $n$.
\end{defn}

Throughout this section, unless stated otherwise the coupling are picked from a product measure of some continuous symmetric distribution whose support is all the real line. In addition, we use notations from the previous section, when handling infinite trees or forests.
Several writers analyzed the spin glass objects in several scenarios, but only few rigorous results were attained. Some of the most impressive results were obtained in \cite{Newman00natureof}, \cite{Newman_Stein_halfplain}. It was shown that GSPs that belong to the support of some $\emph{metastate}$ in the plane (a metastate is a special type of distribution over GSPs), if there is more than one, must agree on all the bonds except for some bi-infinite path. They also showed that in the half plane, there is a unique metastate.
%-------------------------------------------------------------------------------------------------------
\subsection{Main Results}\label{sec:main res-spinglass}
$~$\\
Throughout most of this section we assume that the interactions are distributed according to some nontrivial product measure. Our first result, concerns the geometry of any translation invariant distribution of GSPs (not necessarily a metastate). We show that in any translation-invariant measure for GSPs, no GSP contains an infinite cluster of unsatisfied  (dual) edges. This may be stated as
\begin{theorem}\label{thm:no infty unsat cluster}
In almost every GSP in the support of a translation-invariant measure of GSPs (and interactions) the dual unsatisfied edges do not percolate.
\end{theorem}
Moreover, the collection of unsatisfied (dual) edges forms a forest of positive density. This is one of the first rigorous results about the GSPs of this model.

We then consider the number of GSPs in an infinite graph (when the interactions are distributed according to a product measure of continuous distribution supported on the real line).
We prove that cylinders and deformed cylinders (with some restriction) a.s. have a unique GSP.

On the other hand, we show that regular trees (of degree at least $3$) have infinitely many GSPs, for almost every realization of interactions, and moreover a translation invariant measure of GSPs, supported on uncountably many configurations.

%%-------------------------------------------------------------------------------------------------------
\subsection{Geometry of 2D GSPs}\label{sec:2D-spinglass}
$~$\\
This section will be devoted to exploring the set of unsatisfied (dual) edges in GSPs in the 2D EA Ising spin glass model. More specifically, we restrict ourselves to GSPs which are in the support of some translation-invariant joint distribution of couplings and GSPs. We shall sometimes call such a distribution a translation-invariant $\emph{scheme}$ of GSPs. Throughout this section we consider $\mu_\mathcal{J}$, a translation-invariant scheme of GSPs (where the graph is the square lattice), and that $\mathcal{J}$ is the product measure of a symmetric distribution
which is supported on all the real line.

We begin with a couple of easy lemmas.
\begin{lemma}\label{lem:un-sat form a forest}
For any GSP in the support of $\mu_\mathcal{J}$, the set of dual edges whose primal edges are not satisfied form a forest.
\end{lemma}
\begin{proof}
Indeed, if there was a closed cycle in this dual graph, then flipping all the spins in the finite plane region which is bounded by this cycle would reduce the energy.
\end{proof}

We denote by $\mathfrak{F} = (\mathfrak{V},\mathfrak{E})$ the forest made of the unsatisfied dual edges of some GSP in the support of the measure.

\begin{lemma}\label{lem:the forest has positive dens}
The edges (vertices) of $\mathfrak{F}$ have a positive density, i.e. the limit
$\lim_{n\rightarrow\infty}\sharp(\mathfrak{E}\bigcap{E_n})/n^{2}$ exists and is greater than $0$, where $E_n$ is the
set of edges of a $n\times n$ dual square which has a fixed center (e.g. $(0.5,0.5)$), and there is a similar expression for the vertices.
\end{lemma}
\begin{proof}
Due to ergodic decomposition the existence of the limit is clear. Positiveness will follow if we can show that there are unsatisfied edges a.s. More generally, we show that a.s. there is a positive bound for the fraction of unsatisfied edges in any spin configuration, given the interactions.

Consider a unit square in the lattice. With positive probability the product of the interactions of its edges is negative. But then under any choice of spins, the product of its edges' values is negative, and hence at least one of these edges must be positive.
\end{proof}

We now reach to the main aim of this chapter, proving Theorem \ref{thm:no infty unsat cluster}. This theorem is one of the main results of this work, and its proof is slightly more entangled then previous proofs. We formulate the theorem slightly different than in Subsection \ref{sec:main res-spinglass}, though the formulations are clearly equivalent.
\begin{theorem}\label{thm:no infty unsat cluster}
All the connected components of $\mathfrak{F}$ are finite a.s.
\end{theorem}
\begin{proof}
From the results of Subsection \ref{sec:infinite trees-inv subgraphs} it follows that any infinite connected component must be either a bi-infinite tree or a single-infinite tree.

We begin by showing that no bi-infinite trees components exist in $\mathfrak{F}$.
Without loss of generality. $\mu_\mathcal{J}$ is ergodic. For any bi-infinite component we can define, as in Subsection \ref{sec:infinite trees-inv subgraphs} its $\emph{path}$, the single bi-infinite path which is contained in it. Let $\mathfrak{P}$ denote the union of these paths. It has a well defined density (in the sense of Lemma \ref{lem:the forest has positive dens}), due to ergodicity. We would like to show this density is $0$. Assume the density equals some $\rho > 0$. Consider a large $N\times N$ square in the dual lattice. There are constants $A,B$ s.t. with probability $1$ we can choose $N$ large enough such that such that the sum of absolute values of the squares boundary interactions (of the primal edges) is not more that $AN$, the number of edges of $\mathfrak{P}$ which lay in the square is at least $(\rho/2)N^{2}$ and that the some of these edges interactions, in absolute value is at least $BN^{2}$. Indeed, standard analysis shows that for $N$ large enough, even if we consider the smallest (in absolute value) $(\rho/2)N^{2}$ interactions, their sum will be at least some fixed multiple of $N^{2}$. Moreover, we can choose $N$ to be so large that $AN-2BN^{2}<0$

But now note that these paths divide the square into disjoint regions. each edge of $\mathfrak{P}$ which is in the interior of the square appears in the boundaries of exactly two such regions. The rest of the regions' boundaries are edges from the original square's boundary, each appears exactly once. Thus, the sum of the values over the boundaries (with multiplicity) is at most $AN-2BN^{2}<0$ (as the edges of $\mathfrak{P}$ have all negative value). But this means that there exists at least one region that if we flip all its spins the Hamiltonian decreases. A contradiction.

We now move to the more sophisticated part - showing that there are no single-infinite components.
The crux is the following. We still can decompose the square into regions whose boundaries are either part of the square's boundary or of the trees. But now we do not know whether the sum of boundary values is negative. We do know, due to Theorem \ref{thm:single-infinite at least boundary sized or nlogn}, that the number of edges from the trees that form these boundaries is superlinear (it was shown for a single tree - but the exact same consideration works for a forest). This would suffice if the absolute values of the interactions were bounded away from $0$. But in the continuous case we have to work harder, yet, the result will follow immediately from the next lemma.
\begin{lemma}
Let $\mathcal{J}$ be a product measure of a continuous distribution whose support is all the positive real line on the dual square lattice.
Let $\tau_\mathcal{J}$ be a translation invariant scheme of forests of single-infinite trees (by scheme we mean, again, a joint distribution of couplings and trees).
Then the expected sum of interactions of edges of the forest which lay on a simple path which starts and ends in the boundary of a $N \times N$ square around the origin is a super linear function of $N$.
\end{lemma}
Knowing the lemma guarantees that we can divide the square into regions, bounded by the square's boundary, and the trees' edges, that at least one of the regions has a boundary whose sum of values is negative, and hence a flip reduces the Hamiltonian.
\begin{proof}
Let $u,v$ be two vertices in the forest. Define the mass function $m_t(u,v)$ in the following manner. $m_t(u,v)=1$ if $v\in S(u)$, $|v-u|_\infty = t$ and the interaction of the single edge in the forest that connects $v$ to $S(v)$ is at least $\delta$, where $\delta$ is a constant to be chosen later on. Otherwise $m_t(u,v)=0$. Let $O$ be the "origin" of the dual graph, the point $(0.5,0.5)$. Denote by $E_t$ the value of $\mathbb{E}[\sum_{v\in\mathbb{Z}^{2}}m_t(O,v)|O~is~in~the~forest]$, this value will not change if we replace $O$ by any other vertex $u$, of course. By the MT principle
$E_t = \mathbb{E}[\sum_{u\in\mathbb{Z}^{2}}m_t(u,O)|O~is~in~the~forest]$.
Denote by $p_t$ the probability that for the origin $O$ (or any other vertex), which is conditioned to be in the forest, there is a vertex $u\in R(O)$, at $l_\infty$ distance $t$ from $O$ and that the edge which connects $O$ to $S(O)$ has an interaction at least $\delta$. Note that
\begin{equation}\label{eq:markov eq}
p_t \geq E_t/4t
\end{equation}
This is due to Markov's inequality, as $\sum_{u\in\mathbb{Z}^{2}}m_t(u,O) \leq 4t$ ($4t$ is the total number of vertices at $l_\infty$-distance $t$ from $O$, similarly to Theorem \ref{thm:single-infinite at least boundary sized or nlogn}). We now show that for a small enough $\delta$ there is a constant $C$ which satisfies the following condition
\begin{equation}\label{eq:perc ineq}
\sum_{t\leq s\leq t+C\log{(t+1)}}E_s \geq 1
\end{equation}
Indeed, choose some $\delta$ such that the edges with interaction smaller than $\delta$, if taken as open edges form a subcritical percolation. For a large enough $c$ the probability a given $m\times m$ square contains an open cluster of size greater than $c\log{(m)}$ decreases polynomially in $m$, and the degree of the polynomial in an increasing function of $c$ (which grows to $\infty$ as $c$ does), see \cite{grimmett}. Since the forest has a positive density, if we consider a given $m\times m$ square around a vertex, conditioned this vertex is in the forest, we get a similar bound.

Take some vertex $u$, conditioned to be in the forest. since $S(u)$ is a single-infinite path, if we consider the intersection of $S(u)$ with the annulus $A(u,t,C) := \{v ~s.t.~  t\leq |u-v|_\infty \leq t+C\log{(t+1)}\}$, then there is some connected path in this intersection which connects the internal and external boundary. In particular it length is at least $C\log{(t+1)}$. Due to the above remarks, for large enough $C$, at least two of its edges will have an interaction of at least $\delta$, with very high probability. Thus, $C$ can be taken as such that Inequality \ref{eq:perc ineq} is valid.

Denote by $p_t ^{R}$ the probability that for the origin $O$ (or any other vertex), which is conditioned to be in the forest, there is a vertex $u\in R(O)$, at $l_\infty$ distance $t$ from $O$, and moreover $x(u)-x(0) = t$, where $x$ denotes the $x$-coordinate of the vertex (in words - the distance in achieved in the $x$-axis, from the $\mathbf{right}$), and that the edge which connects $O$ to $S(O)$ has an interaction at least $\delta$. Similarly we define $p_t ^{L}$, $p_t ^{U}$, $p_t ^{D}$ (L stands for left, U - up, D - down). It is trivial that
\begin{equation}\label{eq:trivial eq}
p_t ^{R} + p_t ^{L} + p_t ^{U} + p_t ^{D} \geq p_t
\end{equation}

We now finish the argument in a similar manner to that of Theorem \ref{thm:single-infinite at least boundary sized or nlogn}.
Consider a large $n\times n$ square around $O$, the expected sum of edges whose interaction is at least $\delta$ and that lay on a simple path that is contained in the forest and connects two boundary points of the square, which is a lower bound to the number we are chasing after, is at least
$\Delta := \sum_{\gamma\in\{L,R,U,D\}}\sum_{t = 1}^{t=\lfloor(n-1)/2\rfloor} (n-2t)p_t^{\gamma}$.
The reason is that $e={v,u}$ is an edge (with $u\in S(v)$) which lays on a simple path that connects two boundary edges, iff $R(v)$ intersects the boundary, as $S(v)$ always intersect it.
Due to Equation \ref{eq:trivial eq} $\Delta \geq \Delta_1 := \sum_{t = 1}^{t=\lfloor(n-1)/2\rfloor} (n-2t)p_t$.
Due to Equation \ref{eq:markov eq} we have $\Delta_1 \geq \Delta_2 := \sum_{t = 1}^{t=\lfloor(n-1)/2\rfloor} (n-2t)E_t/4t$.
And due to Inequality \ref{eq:perc ineq} we have $\Delta_2 \geq \sum_{m=1}^{N(n)} (n-2s_m)/4s_m$, where $s_m$ is the sequence that is defined by $s_1=1+C\log 2$, $s_k = s_{k-1} + C\log{(s_{k-1}+1)}$, and $N(n)$ is the last such $m$ with $s_m < n/3$.
It is clear that $\Delta_2 \geq n\sum_{t=1}^{n/(C'\log(n))} (C''t\log(t+1))^{-1}$, and this is superlinear, as $\sum_{m=1}^{\infty}(m\log(m))^{-1})$ diverges. And the result follows.
\end{proof}
Thus, the forest is made only of finite components
\end{proof}

%%--------------------------------------------------------------------------------------------------------
\subsection{Families of Graphs with Unique GSPs}\label{sec:Graphs with Unique GSPs}
$~$\\
In this section we prove two main results. The first is that under a variety of conditions, cylinders and deformed cylinders have only one GSP. The second shows a sufficient condition for planar graphs to have only one GSP. In a future paper we generalize the planarity condition to wider families of graphs. Throughout the thesis we state the results and prove them assuming planarity of graphs, as it is easier both to formulate and to visualize.

\begin{theorem}\label{thm:one ground state in cylinders}
Let $G$ be a connected deformed cylinder that each of its slices is connected as well. Assume that there exist some $N\in\mathbb{N}$ such that there are both infinitely many positive levels and infinitely many negative levels whose slices have the following property:
For each of these slices there are no more than $N$ edges touching it or contained in it. Then $G$ has exactly one GSP.
\end{theorem}
\begin{proof}
Due to what was said above, we shall prove only the uniqueness. Assume, in order to reach a contradiction, that there are at least two GSPs. Remember that a GSP is actually a pair of configurations, which differ only by a global flip. Choose a representative to each GSP. Let these representatives be $\mathcal{C}_1,\mathcal{C}_2$. We can divide the vertices of $G$ into two sets - the set of vertices where $\mathcal{C}_1,\mathcal{C}_2$ agree, and the set where they do not agree. We can divide each of these two sets into connected clusters according to connectivity in the graph. Note that for any such a cluster from one set, all of its neighbors must belong to cluster from the other set, as if there where two different clusters from the same set, with an edge between them, since the division was according to connectivity in $G$, these two clusters must have been the same cluster. An immediate consequence is that an edge has a different value in $\mathcal{C}_1,\mathcal{C}_2$ if and only if it connects two different clusters.
If there is only one cluster, we are done, and there is only one GSP. On the other hand, if there are more than one cluster, each cluster must be infinite. Indeed, if there were a finite cluster, the sum of the values of the edges which connect it to the other clusters (there are such edges, since the graph is connected) must have been non zero (since the coupling distribution is continuous). But then in one of the two GSPs it should have been negative (this amount in $\mathcal{C}_1$ is minus the amount in $\mathcal{C}_2$). Hence, in one of the two ground states we could have reduced the energy by flipping a finite set of spins.

We should now notice that any vertex $v$ belongs to a finite connected subgraph of $G$ which is bounded by two slices, each of them has no more that $N$ edges, each edge with the property that one of its vertices belongs to the slice, where $N$ is the integer from the formulation of the theorem. Indeed, it follows from the assumptions of the theorem that if $v$ is of level $a$, there are integers $b,c$ with $b<a<c$ such that the slices in the levels $b,c$ have the above property. Note that any edge from this subset to its complement in the graph must have one vertex in one of these border slices, and one vertex "outside".

The main observation is that there exist a positive $\varepsilon=\varepsilon(N)$ such that any slice with the above property is fixed, with probability at least $\varepsilon$. Indeed, there are only finitely many isomorphism classes of graph to that slice. Consider such a class. Consider some specific spanning tree of it. Choose an order on its vertices, such that the vertex in the $i^{th}$ place, has all its neighbors (in the tree!) in places $j\leq i$, except at most one neighbor. Now, there is a positive probability such that for any $i$, if it has an edge in the tree, $e$, which is to a neighbor in higher place in the order, then the absolute value of the coupling in this edge is higher than the sum of all the edges which have at least one vertex in the slice, except for those in the spanning tree, together with the sum of absolute values of all couplings of the edges of $i$ in the tree (other than $e$).
But in this case the tree is fixed. As there are only a finitely number of isomorphism classes of the slices' graph (and we needed no information about the rest of $G$, except the bound for the number of edges which touch the slice), there exist a positive $\varepsilon=\varepsilon(N)$ such that any slice with the above property is fixed.

It should be noted that the spanning tree and the order were crucial. If we had not have such an order, we could not have guaranteed a fixed spanning subgraph. If the tree were not spanning, it might have happened that only parts of the graph would have been fixed.

Thus, standard arguments show that with probability $1$ (on the interactions) any vertex belongs to a finite subgraph which is bounded by two fixed slices (and the boundary edges of this subgraph are only those from these two subgraphs to the rest of the graph.
A fixed edge has the same value in any GSP, and the same holds for a fixed slice. We can find an order preserving, injective and surjective map, $S$ from the set of fixed slices to $\mathbb{Z}$, where the order of the slices is induced from the order of their levels.

We remember that any of the above clusters must be infinite. From the last remark it follows that any fixed slice must be contained fully in the interior of a cluster, and there is no cluster which is contained in a finite subgraph which is bounded from both sides by fixed slices (it must be infinite). In addition, due to the connectivity of clusters, the images under $S$ of the fixed slices which are contained in it must be in the form $\{z\in\mathbb{Z}| a<z<b\}$ where $a,b\in\mathbb{Z}\bigcup\{-\infty,\infty\}$.
From here it is immediate that there are at most two clusters. Indeed, if there were three, one of them must have contained only fixed slices whose image under $S$ is of the form $\{z\in\mathbb{Z}| a_0<z<b_0\}$ where $a_0,b_0\in\mathbb{Z}$. But then this cluster is contained in the finite subgraph of $G$ which is bounded between the slices $S^{-1}(a_0),S^{-1}(b_0)$, and in particular be finite.
%!!!!!!!!!!!!!!!!!!!!!!!!!!!!!!!!!!!!!!!!!!!!!!!!!!!!!!!!!!!!!!!!!!!!!!!!!!!!!!!!!!!!!!!!!!!!!!PICTURE!!!!!!!!!!!!!!!!!!!!!!!!!!!!!!!!!!!!!!!!!!!!!!!!!!!!!!!!!!!!!!!!!!!!!!!!!!!!!!!!!!!!!!!!! 
\begin{figure}[h]
\begin{center}
\includegraphics[width=0.7\textwidth]{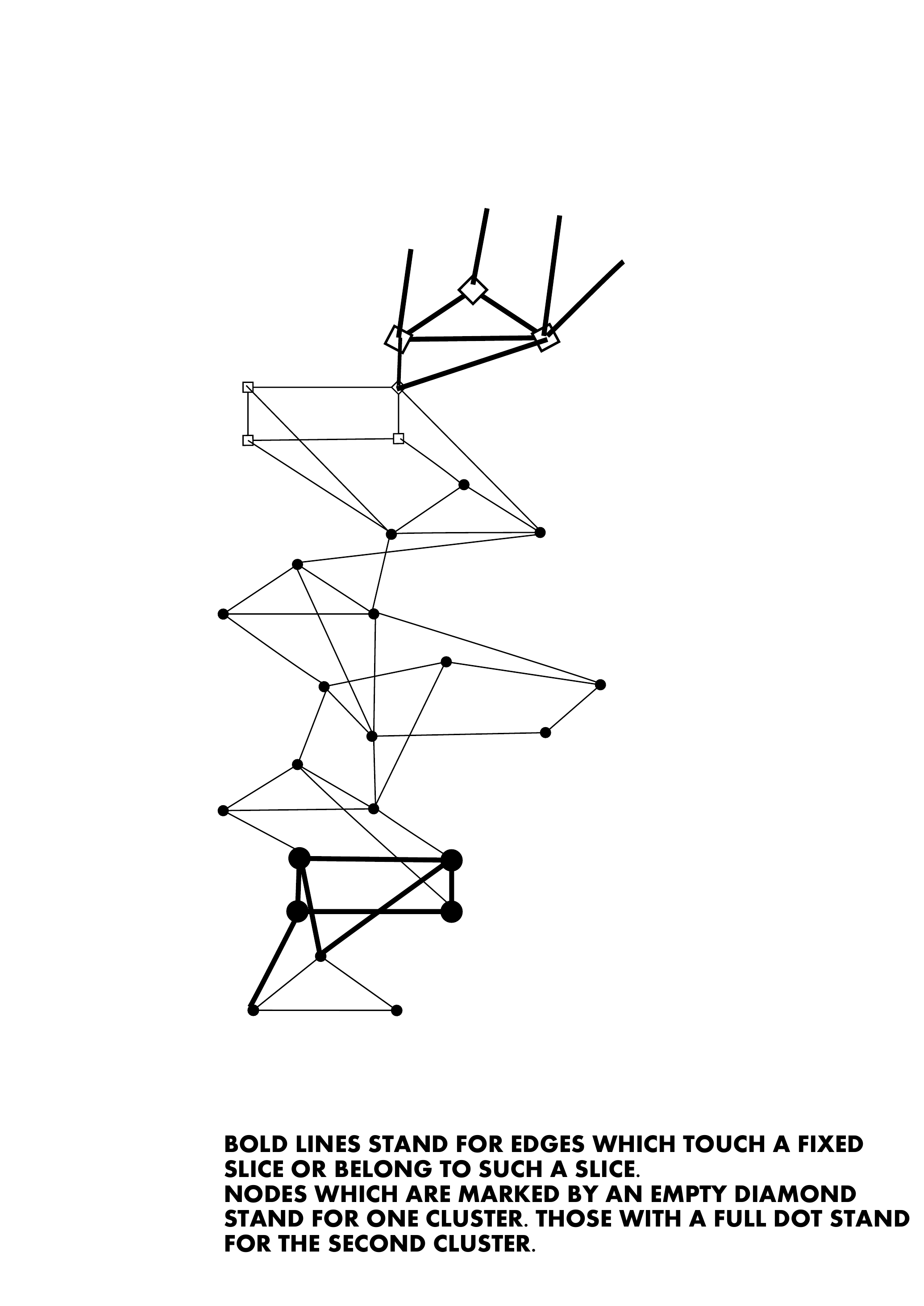}
\caption{Clusters of agreement$\backslash$ disagreement between two GSPs - Theorem \ref{thm:one ground state in cylinders}}
\end{center}
\end{figure}

We are left with case that there are two clusters. Assuming there were two, the sum of values of boundary edges between them (which is finite, as  it belongs to a finite subgraph bounded between two fixed slices) must be positive in one GSP and negative in the other: denote by $h$ the value of that sum in one GSP, the corresponding sum in the other GSP will be $-h$. As the coupling distribution is continuous, $h\neq 0$ a.s. With probability $1$ there is a slice with at most $N$ edges which touches it or belongs to it, and has no edge in common with this boundary, and that each of its edges has an absolute value less then $|h|/(N+1)$, where $h$ is the sum of values of boundary edges in $\mathcal{C}_1$. Assume, without loss of generality. that $h<0$. Note that this boundary and the slice we have just describe are a boundary of some subgraph of $G$. But the sum of the values of this subgraph's boundary edges must be at most $-|h|/(N+1)$, hence, flipping all of this subgraph's spins reduces the Hamiltonian, in contradiction to the assumption that $\mathcal{C}_1$ is a ground state.
\end{proof}
Similar considerations lead to the next result (we formulate it for cylinders, for clarity, though it can be stated for deformed cylinders that satisfy conditions similar to those of the above theorem):
\begin{theorem}\label{thm:no infinite unsat cluster in cylinders}
Let $G$ be a connected cylinder. Then for almost every realization of couplings, in the GSP of $G$ there is no infinite connected cluster of unsatisfied edges.
\end{theorem}
\begin{proof}
$\mathbf{Sketch}$. There is a positive probability that a pair of consecutive slices is fixed, and moreover, all the edges connecting these slices are fixed to be positive. Thus, any vertex belongs to a finite subgraph which is bounded by two such pairs of slices. But an infinite connected cluster which contains this vertex must contain at least on edge from connecting the two slices in one of these pairs. But this edge cannot be negative.
\end{proof}

The next theorem use a similar property to prove uniqueness of GSPs for many planar graphs. This result may be extended to much wider families of graphs.
\begin{defn}
Let $N$ be some fixed integer and let $v$ be a vertex. 

A cycle (a simple closed path) $\mathcal{C}$ is said to $\emph{N-choke}$ (or simply $\emph{choke}$) $v$ if $\mathcal{C}$ surrounds $v$ and the number of edges which lay on it or touch a vertex in the it is smaller than $N$.

A locally finite planar graph $G$ is said to be $\emph{choked}$ if for every vertex $v$ in the graph, there is a positive integer $N=N(v)$, and there is a sequence of cycles in the graph, $\{\mathcal{C}_n\}_{n\in\mathbb{N}}$, where each $\mathcal{C}_{n}$ ($N-$)chokes $v$ and such that for all $n\in\mathbb{N}$,
$\mathcal{C}_{n+1}$ surrounds $\mathcal{C}_{n}$.
%
%A cycle with the property that (with respect to $v,N$) for each of these cycles the number of edges which compose the cycle or touch a vertex in the cycle is smaller than $N$ is said to be $\emph{choking}$ $~v$.
\end{defn}
\begin{rmk}
In the above definition, we used the term 'locally finite graph' for describing a planar graph such that in any compact subset of the plain there are at most finitely many of its vertices. By 'surrounded' we meant that the surrounded object falls into the finite connected portion of the plane after the deletion of the surrounding object.
\end{rmk}

\begin{claim}\label{clm:choked planar have fixed cycles}
In any choked graph $G$, with probability $1$, any vertex is surrounded by a sequence of infinitely many fixed cycles (by that we mean that the edges of the cycle are fixed) such that each cycle surrounds the cycles which come before in the sequence.
\end{claim}
\begin{proof}
Let $v$ be a vertex, $N=N(v)$ be as in the above definition, and $\mathcal{C}$ any cycle $N$-choking $v$.
There is a positive probability, which is bounded from $0$ by a positive function of $N$ that $\mathcal{C}$ if fixed, this can be seen similarly to Theorem \ref{thm:one ground state in cylinders}.
Now standard considerations like Borel-Cantelly finish the argument as the event that $\mathcal{C}_n$ is fixed is independent of the event that $\mathcal{C}_m$ is fixed for far enough $n,m$.
\end{proof}

\begin{theorem}\label{thm:one ground state in choked graphs}
Let $G$ be a choked graph then with probability $1$ it has only one GSP.
\end{theorem}
\begin{proof}
Again, there is at least one, from compactness. Assuming there were at least two, consider as in Theorem \ref{thm:one ground state in cylinders} one representative for each  GSP and the division of the vertices to those which have the same spin under the two representatives and those which have opposite spins. Again we divide into connected clusters, and again each connected cluster must be infinite. Again no fixed edge can connect two different connected clusters. Due to planarity we may consider the dual edges to those which connect different clusters. These must divide the plane into several connected regions, each of which should contain an infinite number of dual faces - since it must contain an infinite number of vertices, and hence no subset of dual edges can form contain a closed cycle. Hence, any dual edge between two clusters lays on a bi-infinite simple path made of dual edges that are between two different clusters. We call such a path in the dual graph a $\emph{domain-wall}$. Note that two domain walls may intersect (the domain walls are not necessarily disjoint paths). Let $v$ be a vertex on an edge whose dual edge, $e^{*}$ is in a domain wall (i.e. or the two representatives agree on $v$'s spin but disagree on the spin of one of its neighbors, or the opposite).
This dual edge, $e^{*}$, is surrounded by some fixed cycle, with probability $1$, by the previous lemma. Thus, any domain wall which contains this dual edge must either intersect the choked cycle or stay only in the finite part surrounded by it. It cannot intersect, since the intersection point would lay on a fixed edge, whose dual is in a domain wall, but this is impossible. The other option is impossible again, as this domain wall must be an infinite simple path, but there are only finitely many points in the region that is bounded by the fixed cycle.
\end{proof}

%%--------------------------------------------------------------------------------------------------------
\subsection{Regular Trees Have infinitely Many GSPs}\label{sec:inf GSPs in trees}
$~$\\
This section explores the nature of spin-glasses over trees. For convenience, we prove our results only for regular trees, although they can be trivially extended to much wider families of trees.
We show that, under our usual model, regular trees of degree at least $3$ have uncountably many GSPs. This is surprising as it is 'clear' there should be only one GSP, the one where every edge has a positive value (in trees, since there are no loops, this can always be done). The only obstruction is that GSP need not to be a global minimizer of energy, but only a local minimizer - it should be 'better' then all the configurations which differ from it by finitely many spin choices. Our main technique is percolation theory for trees.

\begin{theorem}\label{thm:trees have infinitely many GSPs}
Under the model described in Subsection \ref{sec:Prelim-spinglass}, in the d-regular tree, for $d\geq 3$, for almost every coupling there are infinitely many GSPs.
\end{theorem}
\begin{proof}
We begin with an ad-hoc definition. We say that an edge $f$ in an infinite graph is a $\emph{bridge}$, if after deleting it from the graph, but leaving its vertices in the graph, the graph is divided into two connected infinite subgraphs.

Let $T$ be the regular tree, and let $\mathcal{C}$ be one of the two configurations which create the trivial GSP, the one where all the edges have positive value. We do the following:

Take some light edge $e$ (the coupling has very small absolute value) with value $h$. Flip all the vertices in one connected cluster of $T\setminus \{e\}$. The only edge which changed its value was $e$. Because all the other edges have positive value, the new configuration will not be a ground state iff there are some edges $\{e_1,...,e_n\}$ whose sum of values is smaller than $h$, and such that if we delete only them from the tree, we have no finite connected cluster, but if we remove both them and $e$ then there is (exactly one) finite connected cluster. Flipping the finite cluster they bound, reduces the energy.

Because all the other edges have positive values if the sum of the values of $e_i$ is smaller than h, each of them should be smaller than $h$ as well. Now, for $h$ small enough, percolation reasons show that there is a positive probability that the cluster which contains $e$ and only the edges of value at least $h$ is infinite and $e$ is a bridge in it. Indeed, if the degree of vertices in the tree is $d$, it is a common knowledge (e.g. \cite{grimmett}) that a percolation on the $d$-regular tree, with $p>1/(d-1)$, where $p$ is the probability we do not erase an edge from the graph, has infinitely many infinite connected clusters, and in particular, any vertex or edge has positive probability to be contained in one of them. Moreover, for any edge there is a positive probability to be a bridge in the cluster. For $h$ small enough, the probability an edge coupling has absolute value higher than $h$ is as high as we wish, and from the above it is clear there is a positive probability that in the graph obtained from deleting from the tree all the edges of interaction smaller than $h$ (except, maybe, the edge $e$), $e$ is a bridge in an infinite connected cluster.

In this case, sets like $\{e_1,...,e_n\}$ with the property described above do not exist. Hence, if the described cluster happens to be infinite, $e$ happens to be a bridge in it, and we flip all the spins in one side of the edge $e$, we get a new ground state.

There are infinitely many edges where the described event happens, and hence there are infinitely many GSPs (which are clearly different).
\end{proof}

With similar techniques one can even show the claim below.
\begin{theorem}\label{thm:trees have uncountably many different GSPs}
Under the model described in Subsection \ref{sec:Prelim-spinglass}, in the d-regular tree, for $d\geq 3$, for almost every coupling there are uncountably many GSPs.
\end{theorem}
This is not surprising, as there is not much difference between Ising spin-glass on trees and ferromagnetic Ising on trees. And even in $\mathbb{Z}^{d}$-lattices,
in the Ising ferromagnetic model, there are infinitely many GSPs. On the other hand, we show below an important difference between the ferromagnetic
Ising on trees and on lattices: There is a translation-invariant measure, supported on uncountably infinitely many GSPs, for trees, while for lattices (and more generally for Cayley graphs of amenable groups)
in any translation invariant measure in supported on one single GSP. The latter result will appear again in Section \ref{sec:dynamics-loops} when we discuss the Loop dynamics.

We shall use the following lemmas a couple of times in this work:
\begin{lemma}\label{lem:very biased perc cofinite}
Let $T$ be a $d$-regular tree, for some $d\geq 3$. Then there is some $1>p>0$ s.t. in the Bernoulli($p$) vertex (edge) percolation the complement in $T$ of the union of infinite clusters is made of finite clusters only.
\end{lemma}
This is of course much stronger than saying that there are infinite clusters in the percolation.
\begin{proof}
$\mathbf{Sketch.}$ We prove for site percolation, minor changes apply for bond percolation as well. We consider the percolation as coloring. A vertex is colored black with probability $p$ and otherwise white. The proof is based on the following simple facts that we do not prove:
%TO DO IT IN THE FUTURE!!!!!!!!!!!!!!!!!!!??????????!!!!!!!!!!!!!!!!!!
\\ I. Given that a black vertex is in a finite cluster or that it is in a finite cluster and has (at least) one white neighbor, the expected size of this cluster is bounded by $1+\alpha (p)$, where $\alpha(p)\rightarrow 0$ as $p\rightarrow 1$ is a nonnegative, continuous function of $p$.
\\ II. Given that a vertex is white, the expected size of the white cluster containing it is (for $p$ close enough to $1$) $1+\beta (p)$, where again $\beta \rightarrow 0$ as $p \rightarrow 1$.
\\ III. Given that a black vertex $v$ has one white neighbor, the probability that the black component which contains it is finite is $\gamma (p)$, where $\gamma$ is a nonnegative function of $p$ which tends to $0$ as $p$ tends to $1$.

Assume for example that a given vertex $v$ is black and lays in a black finite component. Then the size of this component is expected to be at most $1+\alpha$ (we do not write $p$ as it is fixed). All the neighboring clusters are white. There are at most (and actually, less than) $d ( 1+\alpha (p) )$ such clusters. The expected total sum of their sizes is no more than $d ( 1+\alpha (p) )(1 + \beta )$. And the last number is also a bound for the number of neighbors all these clusters have (together). All these neighbors are black, and hence, the expected number of finite clusters among these is no more than $d ( 1+\alpha (p) )(1 + \beta ) \gamma$. This number tends to $0$ as $p \rightarrow 1$, and hence, starting at some point, it is less then $1$. If we now continue the same analysis with any of the black neighbors which are parts of finite black components, we get a similar bound for the number of finite black components which are neighbors of the white components which are neighbors of this new black vertex (and we do not count clusters we have already reached), and we can continue like this ad infinitum. But what we have shown is that the complement of the infinite clusters is stochastically dominated by a sub-critical Galton-Watson process, and hence finite a.s. Only slight changes are needed if we start with a white vertex instead of black.
\end{proof}
\begin{lemma}\label{lem: very biased - even majority cofinite}
Consider a Bernoulli $p$ site percolation on a $d$-regular tree $T$ as a coloring s.t. with probability $p$ a vertex is colored black and otherwise white (independently of the others). Color by red the vertices that the majority of their neighbors and them are colored black. Then if $p$ is close enough to $1$, the graph which is the result of removing the red vertices from $T$ (and edges which touch them) has only connected clusters of finite size.
\end{lemma}
\begin{proof}
$\mathbf{Sketch.}$ In \cite{LiggettStaceySchonmann} it is shown that if we color vertices in two colors according to the product measure of Bernoulli(p) distributions and then recolor according to vertices which the majority of their neighbors and them is in some given color, then for any $1>q>0$, there is a $1>p>0$ s.t. the resulting recoloring, considered as site percolation, stochastically dominates the Bernoulli $q$ site percolation (Actually they show there something much more general then this majority process).
Combining this with the above lemma finishes the argument.
\end{proof}

\begin{theorem}\label{thm:translation invariant GSPs in trees}
Let $T$ be a $d$-regular tree for $d\geq 3$, and consider the model of Subsection \ref{sec:Prelim-spinglass}, where the distribution of the interactions is an i.i.d. continuous  probability measure
whose support is $\mathbb{R}_+$. Then there exist a translation-invariant measure $\mu$ supported on uncountably many GSPs.
\end{theorem}
\begin{rmk}
Of course, the positivity of the interactions is not important.
\end{rmk}
\begin{proof}
$\mathbf{Sketch}$.
First note that knowing which edge is satisfied is equivalent to knowing the GSP.

There exists some small $\varepsilon > 0$ with the following property: Given an edge $e$ with interaction smaller than $\varepsilon$ in the tree, there is only a finite number of finite connected subgraph of the tree s.t. $e$ is in their boundary and less than half of the boundary edges have interaction smaller than $\varepsilon$. Indeed, there is an exponential number of connected subgraphs of the tree containing $e$ as a boundary edge having some given side. For small enough $\varepsilon$ the probability that one of them has more than half of its boundary edges having interaction smaller than $\varepsilon$, is exponentially small, with exponent base as small as we wish. In particular, all these probabilities are summable for small enough $\varepsilon$. A Borel-Cantelly argument shows that there are indeed only a finite number of finite connected subgraphs of the above type.

Moreover, if $\varepsilon$ is so small such that the above sum (which is actually the expected number of "bad" finite connected subgraphs) is smaller than $1$, then with probability $1$ there are edges $e$ with no finite connected subgraphs of the above type.
For any such edge flip a coin whether this edge is made satisfied or not.

It is easily seen that we have created a translation invariant measure of GSPs. It is indeed a GSP as in every connected subgraph whose boundary contains an unsatisfied edge, most boundary edges have interaction at least $\varepsilon$. The remaining edges, even if all of them are flipped are dominated both in their number and in their sum (in absolute value).
\end{proof}

\begin{theorem}\label{thm:single GSP in Z^d if J=1}
Let $\mu$ be a translation invariant measure of ground states of the EA Ising spin glass system over $\mathbb{Z}^{d}$, with all the interactions being positive, then
$\mu$ is supported on a single GSP, the one with all the spins the same.
\end{theorem}
\begin{proof}
$\mathbf{Sketch}$. Assume there is another different GSP in the support of $\mu$, let $\sigma$ be one of the two representatives for it.
In the dual complex $\mathbb{Z}^{d}*$, construct the complex made of pieces of codimension $1$, which distinguish between areas of positive spins to those of negative spins in $\sigma$.
By a dual piece of codimension $1$ we mean the hypercube of codimension $1$ and edges of length $1$, perpendicular to some edge of the graph $\mathbb{Z}^{d}$ and passes through the middle of that edge.
We call those pieces which distinguish between vertices of different spin $\emph{domain walls}$.
Consider a large box of side $N$ around the origin. Translation invariance implies that inside that box there are $\Theta(N^{d})$ pieces of the domain walls.
Note that any such piece corresponds to an edge whose weight is negative (its two spins are opposite). The sum of all of them is expected to be $\Theta(N^{d})$ as well.
The sum of all the pieces of the boundary of the box (not only those that may happen to be in domain walls)
is $\Theta(N^{d-1})$. The domain walls divide the box into domains of same sign, such that any two neighboring (by a piece of the wall) domains have opposite spins.
Any piece of the domain wall which is inside the box is hence counted exactly twice as boundary of a domain.

But it follows that the sum of the weights of the boundaries of the domains (including the boundaries of the box, which are counted once) is negative. Thus, The Hamiltonian
of one of these domains must be positive. Flipping all the spins in that domain reduces the Hamiltonian, in contradiction to $\sigma$ being a ground state.
\end{proof}
\begin{rmk}
Actually one can achieve the above result for any amenable Cayley graph.
\end{rmk}
\newpage
\section{Zero Temperature Glauber Dynamics in Statistical Mechanics Problems}
%Dynamics in Statistical Mechanics Problems}
$~$\\
%%--------------------------------------------------------------------------------------------------------
\subsection{Background}\label{sec:dynamics-background}
$~$\\
Dynamic evolution in statistical mechanics has become a subject of interest in the last decades. A description of the system at equilibrium and estimates on the time until such equilibrium is reached are of paramount importance here as it appears in many other scientific areas (e.g. biology, computer science, economics) where dynamical models are relevant as well. Common among such models in statistical mechanics, is the so-called $\emph{Glauber Dynamics}$, which can be used to describe the evolution in many spin-systems with a Boltzmann-Gibbs law at stationarity (such as Ising, Potts, etc.). This part of the work is devoted to obtaining new results about Glauber dynamics for the zero temperature Ising model on several infinite graphs. We introduce results and techniques which extend earlier works of other authors, e.g. \cite{NandaNewmanStein},\cite{GandolfiNewmanStein} \cite{CamiaSantisNewman}.

%%--------------------------------------------------------------------------------------------------------

\subsection{Preliminaries}\label{sec:dynamics-prelim}
$~$\\
Let $G$ be a graph. $\mathcal{X}$ = $\{-1,1 \}^{G}$ will denote the
space of all configurations (spin choices) over $G$.
The Ising model on $G$ corresponds to a Gibbs distribution on configurations
with Hamiltonian
\begin{equation}\label{eq:Hamiltonian}
\mathcal{H}(\sigma) = -\sum_{x\thicksim y} \sigma_x\sigma_y
\end{equation}
where summation is over all pairs of adjacent vertices. The $\emph{Glauber Dynamics}$,
for the Ising model on $G$ in zero temperature is a Markov processes $\sigma_{t}$,
on $\mathcal{X}$. Its evolution is governed by independent Poisson clocks with
rate $1$ which are assigned to every vertex $x \in G$. When such a clock rings a spin flip is $\emph{considered}$.
If the change in the energy
\begin{equation}
\Delta\mathcal{H}_x(\sigma) = 2\sum_{y\thicksim x} \sigma_x\sigma_y
\end{equation}
is negative (positive) a spin change in $x$ is done with probability 1 (0). In the case that the energy change is zero, a coin is tossed in order to decide whether to perform a flip or not. All coin toss outcomes are independent and independent of the initial configuration and the ring times. A flip which changes the energy is called an $\emph{energy-reducing flip}$.

In other words, the configuration process $\{\sigma_t(v) ;\; t \geq 0, v \in V\}$ is almost surely completely determined, given the following:
\begin{itemize}
\item $\sigma_0(v) ;\; v \in V$ - initial configuration
\item $\tau_k(v) ;\; k \geq 1, v \in V$ - rings times.
\item $\beta_k(v) ;\; k \geq 1, v \in V$ - coin toss outcomes.
\end{itemize}
For convenience we will assume that our sample space $\Omega_V$ is the set of all possible assignments to these values. For instance, one can identify $\Omega_V$ with the space
\[
\big( \{-1,+1\} \times \BR_+^\BN \times \{-1,+1\}^\BN \big)^V
\]
and $\varpi = (\sigma_0(v), \tau_k(v), \beta_k(v) ; \; k \geq 1, v \in V) \in \Omega_V$. We denote by $\omega = (\tau_k(v), \beta_k(v) ; \; k \geq 1, v \in V)$. The standard Borel sigma algebra on $\Omega_V$ will be denoted $\calF_V$. If $U \subset V$ we shall treat $\calF_U$ as a subset of $\calF_V$.

A collection $\{\alpha(v) ; \; v \in V\}$ will often by viewed as a function $\alpha$ on $V$. If $U \subset V$, we shall write $\alpha|_U$ for its restriction to $U$. If $\alpha_1, \alpha_2$ are functions on $U$ and $V \setminus U$ resp., we shall write $\alpha_1 \times \alpha_2$ for the function on $V$ which is given by their superposition. We shall employ this notation ``point-wise'' for time-evolving functions on $V$ (e.g. $\sigma_t$) as well.

In the presence of boundary conditions on $U \subset V$, sites $v \in U$, have their values set to a prescribed value
\[
\sigma_t(v) \equiv \pm 1 \; \forall t \geq 0
\]
and a prescription of these will be given as a configuration $\xi \in \{-1, +1\}^U$.

Given the graph $G=(V,E)$ and possibly boundary conditions $\xi$ on $U \subset V$, the function which maps an $\varpi \in \Omega_V$ to the evolution it generates $(\sigma_t)_{t \geq 0}$ will be denoted by $\Sigma_G^\xi$.
\[
    (\sigma_t)_{t \geq 0} = \Sigma_G^\xi(\varpi)
\]
It will be denoted $\Sigma_G$ in the absence of boundary conditions.
\begin{rem}
$\Sigma_G$ is in fact well defined only on a subset of $\Omega_V$, one which has measure one. This can be done using percolation techniques, and may be found in classical textbooks. We prove a much more general argument in Section \ref{sec:dynamics-loops}.

is standard (percolation arguments).
\end{rem}
\begin{rem}
It is more convenient to have the $k$-th coin toss outcome at vertex $v$, i.e. $\beta_k(v)$ be associated with the $k$-th ring at that vertex. Thus at the $k$-th ring, if a there's a tie $\beta_k(v)$ is used and otherwise it is discarded.
\end{rem}

$\mathbf{P}_\omega$ denote the probability distribution on the realizations $\omega$ of the ring times and the coin tosses, and by $\mathbf{P}_{\sigma_{0}}$ the distribution of the initial configuration. Unless stated differently we assume that $\mathbf{P}_{\sigma_{0}}$ is just the product measure of symmetric Bernoully variables. The joint distribution is $\mathbf{P}_\varpi = \mathbf{P}_{\sigma_{0},\omega}=\mathbf{P}_{\sigma_{0}}\times\mathbf{P}_\omega$.
%\footnote{Oren: Make this relate to the previous definitions of $\omega$ and $\varpi$. Also move to next to these definitions.}

We say that the dynamics over $G$ $\emph{freezes}$ if
%with $\mathbf{P}_{\sigma_{0},\omega}$-probability $1$
the limit $\sigma_{\infty}(\sigma_{0},\omega)$ = $\lim_{t\rightarrow\infty}\sigma_{t}(\sigma_{0},\omega)$ exists pointwise. In other words, each vertex flips finitely many times.
If the above limit does not exist, but
%with probability 1 (w.r.t $\mathbf{P}_{\sigma_{0},\omega}$)
some vertices do freeze (flip at most finitely many times) and those which do not freeze do not form infinite connected clusters, we say that the dynamics $\emph{almost freezes}$.  When some vertices freeze but those which do not freeze form infinite clusters we say that the dynamics $\emph{partially freezes}$.
We say that a subset of vertices is $\emph{possibly self freezing (PSF)}$ if with probability $1$ no flip occurs in any of its vertices when they are initially all set to the same value (either all $+$ or all $-$).
A cylinder is said to $\emph{freeze}$ $\emph{in}$ $\emph{slices}$ if it freezes
and in the limit $\sigma_{\infty}$ all vertices in each slice have the same spin value.
A vertex which stops flipping after some finite time is said to be, after that time $\emph{frozen}$.
%\end{defn}
%%--------------------------------------------------------------------------------------------------------
\subsection{Main Results}\label{sec:dynamics-main res}
$~$\\
In \cite{NandaNewmanStein} it is shown that on many families of graphs the dynamics freezes a.s., when the initial configuration is i.i.d. symmetric $\pm 1$. In particular this includes all odd-degree regular trees. Our first result shows that this is not true for regular trees with even degree, i.e. with positive probability a site will flip forever. We also discuss the case of a biased i.i.d. initial configuration.

Next we examine other graphs and show that under mild assumptions, for any bounded degree graph $G$ there exist some $N \in \mathbb{N}$
such that for any $n>N$, $G\times K_n$ freezes almost surely.

We end this section with results and constructions which are related to cylinders. We show that any cylinder (with finite connected slices) almost freezes a.s. We also show that  the $C_n$-cylinder and the $K_n$-cylinder freeze a.s. and moreover - freeze in slices a.s. We show, on the other hand that there are cylinders which do not freeze and some which freeze but not in slices (even when the slice is a transitive graph).

%--------------------------------------------------------------------------------------------------------
\subsection{Glauber Dynamics on Regular Trees with Even Degrees}\label{sec:dynamics-reg trees}
$~$\\
Here we investigate the zero-temperature Glauber dynamics on trees. It follows from \cite{NandaNewmanStein} that this dynamical process on a regular tree of odd degree freezes a.s.
We show the complementary result that a.s. the dynamics does not freeze when the tree is regular of even degree (by that we mean that at least some vertex flips indefinitely),
as long as our initial spin distribution is the product measure of symmetric $\pm 1$ distributions. We show that if the initial distribution is i.i.d. but non-symmetric, with probability $1$ - either the dynamics does not freeze, or it freezes to
a configuration where all the spins are the same.
In addition, we show that in the above scenario, if the initial distribution is biased enough, a.s. the dynamics freezes to the same-spin configuration.
It is worth mentioning that a recent work of \cite{CAVITY} shows that if instead of Poisson clocks, we work in the simultaneous scenario (where all the clocks ring together at integer multiples of time units),
then even if the bias in the initial spin distribution is small, the dynamics freezes to the same-spin configuration. This leads to the next conjecture (supported by simulations in \cite{CAVITY}).
\begin{conjecture}
If the distribution of the initial spin configuration is the product distribution of biased Bernoulli variables then the zero-temperature Glauber dynamics on a regular tree of even degree freezes a.s. to a same spin configuration.
\end{conjecture}

The most inovative result of this subsection is the following.
\begin{theorem}\label{thm:even-regular trees do not freeze}
Consider the zero-temperature Glauber dynamics on a $d$-regular tree where $d$ is even.
Assume that $\mathbf{P}_{\sigma_{0}}$ is the product measure of symmetric $\pm 1$ random variables.
Then with $\mathbf{P}_{\sigma_{0},\omega}$-probability $1$ the dynamics does not freeze, i.e. a.s there exist a vertex which flips its spin indefinitely.
\end{theorem}

Before proving this theorem we shall need several straight-forward propositions. We use the standard order for vectors and functions.
\begin{prop}[Monotonicity]
\label{prop:Monotonicity}
For any $G=(V,E)$, $U \subset V$ and $\xi \in \{-1, +1\}^U$.
\begin{itemize}
\item
    $\Sigma_G^\xi$ is monotone increasing in $\sigma_0$, $(\beta_k)_{k \geq 1}$.
\item
    If $\xi_0 \geq \xi$ then $\Sigma_G^{\xi_0} \geq \Sigma_G^\xi$.
    and if $\xi \equiv 1$ on $U$ then $\Sigma_G^\xi\geq \Sigma_G$.
\end{itemize}
\end{prop}

If $G=(V,E)$ and $U \subset V$, then $G \setminus U$ is the graph obtained from $G$ by removing all vertices in $U$ and any edges the connect such vertices. For $v \in V$, we denote by $G(v)$ the connected component in $G$ which contains $v$. We shall write $G_U(v)$ as a short for $(G \setminus U)(v)$. If $G_0 = (V_0, E_0)$ is a sub-graph of $G$ then its closure $\overline{G_0}$ (with respect to $G$) is obtained from $G_0$ by adding all neighbors to vertices of $G_0$ which are not in $V_0$ and the edges that connect them. The following is immediate.
\begin{prop}[Independence]
\label{prop:Independence}
Let  $G=(V,E)$, $U \subset V$ and $\xi \in \{-1, +1\}^U$.
Suppose that $G \setminus U$ breaks into $k$ connected components $G_1 = (V_1, E_1), \dots G_k = (V_k, E_k)$. Then
\[
    \Sigma_G^\xi (\varpi) =
        \Sigma_{\overline{G_1}}^\xi (\varpi|_{V_1}) \times \dots \times
        \Sigma_{\overline{G_k}}^\xi (\varpi|_{V_k})
\]
where '$\times$' has the obvious meaning.
\end{prop}

\begin{proof}[Proof of Theorem~ \ref{thm:even-regular trees do not freeze}]
Let $x \in V$. If everything freezes with probability $1$, then with positive probability $\sigma_t(x) = 1$ for all $t \geq t_0$, where $t_0 > 0$ is a deterministic number.
From Proposition~\ref{prop:Monotonicity}, for all $\varpi$ such that this happens, also $\sigma^{y+}_t(x) = 1$ for all $t \geq t_0$, where
\[
(\sigma^{y+}_t)_{t \geq 0} = \Sigma_G^{\xi^{y+}} (\varpi),
\]
$\xi^{y+}$ is the boundary condition $\xi^{y+}(y) = +1$ and $y$ is some neighbor of $x$.

Now enumerate all the neighbors of $y$ as $x_1, \dots, x_{d-1}$ and $z$. Since we have just shown that events $\{\sigma^{y+}_t(x_i) = 1
; \; \forall t \geq t_0\}$ happen with positive probability and since they are
$\calF_{G_y(x_i)}$-measurable by Proposition~\ref{prop:Independence}, and therefore independent, the following has a positive probability as well
\[
\CA^+_z(y) = \left \{
    \begin{array}{l}
        \sigma^{y+}_t(x_i) = 1; \; \forall t \geq t_0, i=1, \dots d-1   \\
        \tau^1(y) > t_0 ,\, \sigma_0(y) = +1
    \end{array}
\right\}
\]
But $\CA^+_z(y) \subseteq \{\sigma_t(y) = +1 ;\; \forall t\geq 0\}$.
Event $\CA^-_z(y) \subseteq \{\sigma_t(y) = -1 ;\; \forall t\geq 0\}$ is defined similarly.

Now suppose $z$ has neighbors $y_1, y_2, \dots y_d$. Since $\CA^\pm_z(y_i)$ are
$\calF_{G_z(y_i)}$-measurable, they are independent and therefore we can have  $\CA^+_z(y_i)$ for all $i \leq d/2$ and
$\CA^-_z(y_i)$ for all $i > d/2$ occur at the same time with positive probability. But then $z$ must flip infinitely many times, as its neighbors are constantly in a tie.
\end{proof}

\begin{theorem}\label{thm:tree-non symmetric case}
Consider the zero-temperature Glauber dynamics on a $d$-regular tree for even $d$. Assume that $\mathbf{P}_{\sigma_{0}}$
is the product measure of strictly-biased $\pm 1$ random variables. Then either the dynamics does not freeze with probability $1$ (a.s there exist a vertex which flips its spin indefinitely) or with probability $1$ it freezes to a same spin configuration.
\end{theorem}
\begin{proof}
$\mathbf{Sketch.}$ The above proof shows, in fact, that if we assume a.s. freezing then either a.s. all the vertices freeze to $+1$ or a.s. all the vertices freeze to $-1$. Otherwise we could again find a vertex with half of its vertices frozen to $+1$ and the other half frozen to $-1$. Such a vertex would flip infinitely many times. If the bias is towards $+1$, then by
monotonicity and standard coupling arguments, it is easy to see that the probability that all the vertices freeze to $+1$ is at least the probability they freeze to $-1$. Hence, in this case, either the configuration does not freeze a.s., or a.s. it freezes to the uniform configuration where all the spins are $+1$.
\end{proof}

\begin{theorem}\label{thm:strongly biased freezes}
Consider the zero-temperature Glauber dynamics on a $d$-regular tree for even $d$.
Then there exist some positive $p=p(d), 0<p<1$ such that if $\mathbf{P}_{\sigma_{0}}$
is the product measure of i.i.d. variables with probability $p$ for $+1$, then with
$\mathbf{P}_{\sigma_{0},\omega}$-probability $1$ the dynamics freezes to
the all $+1$ configuration.
\end{theorem}
\begin{proof}
$\mathbf{Sketch.}$
There is a positive $p_1<1$ such that if $p > p_1$, then with
$\mathbf{P}_{\sigma_{0}}$-probability $1$, any region of spins $-1$ is finite
at time $t=0$. Moreover, there exist an even larger $p_2$, such that if $p > p_2$
then the complement of the set of all vertices
whose spin is initially $+1$ and more than half of their neighbors are $+1$
has only finite connected components. This is due to Lemma \ref{lem: very biased - even majority cofinite}.
Clearly, any vertex in this $+1$ set will be $+1$ for all $t \geq 0$.
On the other hand, any finite cluster in the complement is surrounded by
by vertices whose spin is $+1$ for all $t \geq 0$ and therefore after finite time
will have all its vertices $+1$ forever. Indeed, all the leaves of such components must
turn permanently $+1$ after finite time and therefore, by induction, also the entire
component.
\end{proof}
%--------------------------------------------------------------------------------------------------------
\subsection{Glauber Dynamics on Product Graphs}\label{sec:dynamics-product}
$~$\\
We begin this section with some easy lemmas.
\begin{lemma}\label{lem:finite energy reducing}
Let $G$ be a finite graph, $\sigma_{0}$ any initial spins configuration,
then $\mathbf{P}_\omega$-a.s. there is a time $T<\infty$, after which
there are no more energy-reducing flips.
\end{lemma}
\begin{proof}
Any flip can either reduce the total energy of the system $\mathcal{H}(\sigma_t)$
or leave it fixed. As $G$ is finite the total energy of the system is bounded from
below and moreover any energy-reducing flip decreases the Hamiltonian by
at least $2$. Therefore a.s. $\lim_{t \to \infty} \mathcal{H}(\sigma_t)$
exists and moreover is achieved after some finite time $T < \infty$.
\end{proof}

\begin{conclusion}\label{con:almost freezing finite energy reductions}
Under a dynamics which is almost-freezing with probability 1,
any vertex has finitely many energy-reducing flips.
\end{conclusion}
\begin{proof}
Indeed, for any vertex $v$ which does not freeze there is a time $t=t_0$
such that the connected cluster of vertices which flip after this time
and contains $v$ is finite. Clearly until this time $v$ had only finitely many
flips. On the other hand, similar considerations to the ones which appeared in
Lemma \ref{lem:finite energy reducing} show that $v$ will have only finitely many
energy-reducing flips from now on. The result follows.
\end{proof}

We shall now quote a simple lemma whose proof can be found in \cite{GandolfiNewmanStein}.
Let $Z_t$ be a continuous-time time-homogeneous Markov process with state space $\mathcal{Z}$ and let ${Z_t}^{(\tau)}$
denote the time shifted process $Z_{t+\tau}$. If $A$ is a measurable subset of $\mathcal{Z}$, we say that $A$ $\emph{recurs}$ if
\begin{equation}
\label{def:recur1}
\{\tau > 0 | Z_\tau \in A\}    \qquad \text{is unbounded}.
\end{equation}
If $B$ is an event, measurable with respect to the $\sigma$-field generated by $\{Z_t | 0 \leq t \leq 1\}$, we say that $B$ \emph{recurs} if
\begin{equation}
\label{def:recur2}
\{\tau > 0 | Z^{(\tau)}\in B\} \qquad \text{is unbounded}.
\end{equation}
Then,
\begin{lemma}\label{lem:from newman}
If $\inf_{z \in A} \mathbf{P}(B|Z_0 = z) > 0$ then if $A$ recurs with positive probability, so does $B$.
\end{lemma}

Given $\sigma$, a configuration on $G=(U,E)$ and $V \subseteq U$, we denote by $V^+(\sigma)$ the set of vertices whose spin is $+1$ and $V^-(\sigma)$ the set of vertices whose spin is $-1$. In the next definition $d$ is any positive constant.

\begin{defn}\label{def:strongly freezing}
A family of finite graphs $\{G_n = (V_n, E_n)\}_n$ is
%\footnote{Oren: Why the different font for the graphs?} will be called
$d$-$\emph{strongly freezing}$ if for any $0<p<1$ there exists $n_0$ such that for all $n>n_0$, if $\sigma_n$ is a configuration on $V_n$ chosen according to
an i.i.d. symmetric $\pm 1$ product distribution, then with probability at least $p$ either all the vertices in $V_n$ have at least $d+1$ more neighbors in $V^+_n(\sigma_n)$ than in $V_n^-(\sigma_n)$ or all the vertices in $V_n$ have at least $d+1$ more neighbors in $V^-_n(\sigma_n)$ than in $V^+_n(\sigma_n)$.
\end{defn}

\begin{defn}\label{def:0-strongly freezing}
We say that a family of finite graphs $\{G_n\}_n$ is $\emph{strongly freezing}$ if it is $\emph{d-strongly freezing}$ for all $d\in\mathbb{N}$.
\end{defn}

\begin{exam}
Let $\{K_n\}$ be the sequence of complete graphs of $n$ vertices. $\{K_n\}$ is a d-strongly freezing family for any $d$.
\end{exam}

The following theorem shows that graph-products yield graphs which are more "regular" for the Glauber dynamics. For its proof, we use percolation techniques together with some of the above lemmas.
\begin{theorem}\label{thm:graph products freeze}
Let $\{G_n\}$ be a strongly freezing family of graphs. Then, for any $d\in\mathbb{N}$, there exist $N\in\mathbb{N}$ such that for any graph $G$ of bounded degree $d$ and countably many vertices, the zero-temperature Glauber dynamics over $G \times G_n$ almost freezes a.s. for all $n\geq N$.
Moreover, at any vertex there are at most finitely many energy-reducing flips.
\end{theorem}
\begin{rmk}
This result may be generalized to graphs with a small number of vertices of higher degrees.
\end{rmk}
\begin{proof}
It is well known (e.g. \cite{grimmett}) that for any graph $G$ of bounded degree $d$ and countably many vertices, the critical value for Bernoulli site percolation is at least $1 / (d-1)$. In particular, if the probability of ``opening'' a site is at most $1 / d$, then a.s. there are no infinite open cluster of sites. We pick $p=1 - 1 / d$, and find the $n_0$ as guaranteed from Definition~\ref{def:strongly freezing} for this $p$. Let $n>n_0$ be some integer.
Let $U$ be the vertices of some slice of the graph. By definition, with probability $p$,
either all the vertices in $U$ have at least $d+1$ more neighbors in $U^+(\sigma_{0})$ than in $U^-(\sigma_{0})$ or all the vertices in $U$ have at least $d+1$ more neighbors in $U^-(\sigma_{0})$ than in $U^+(\sigma_{0})$. Therefore these vertices will eventually freeze. Indeed, if, without loss of generality, any vertex in this slice has at least $d+1$ more neighbors in $U^+(\sigma_{0})$ than in $U^-(\sigma_{0})$, then clearly the vertices of $U^+(\sigma_{0})$ will never change their spin. Moreover, if we wait enough time, all the vertices of $U^-(\sigma_{0})$ will change their spin to $+1$, regardless of what are the spins of their other neighbors.
We call a slice with the above property $\emph{good}$ and otherwise $\emph{bad}$. Due to the choice of $p$ bad slices form only finite clusters. Moreover, all the edges from these clusters to the rest of the graph are to good slices. Thus, for any cluster of bad slices, there is some time $T$ after which any edge from this cluster to its complement is to a frozen node. This implies. almost freezing a.s.
%\footnote{Oren: Do you really need to say that all edges from bad clusters to their complements are to frozen nodes?}
The second part of the theorem follows from Conclusion \ref{con:almost freezing finite energy reductions}.
\end{proof}

\begin{conclusion}
For any $d\in\mathbb{N}$, there exist a $N\in\mathbb{N}$ such that for any graph $G$ of bounded degree $d$ and countably many vertices, the zero-temperature Glauber dynamics over $G\times K_n$ almost freezes a.s. for all $n\geq N$
Moreover, at any vertex there are at most finitely many energy-reducing flips.
\end{conclusion}
%--------------------------------------------------------------------------------------------------------
\subsection{Glauber Dynamics on Cylinders}\label{sec:dynamics-cylinder}
$~$\\
In this subsection we treat cylinder-type graphs.

\begin{theorem}\label{thm:cylinders almost freeze}
Let $G$ be a finite graph with no isolated vertices (and at least two vertices). If $\mathbf{P}_{\sigma_{0}}$ is a product measure of i.i.d $\pm 1$ variables, then with $\mathbf{P}_{\sigma_{0},\omega}$-probability $1$ the $G$-cylinder almost freezes. Moreover, any vertex has only finitely many energy-reducing flips.
\end{theorem}
\begin{proof}
Without loss of generality we may assume the graph is connected.
With $\mathbf{P}_{\sigma_{0},\omega}$-probability 1 for any integer $l$ there are two other integers $i,j$ such that $i+1<l<j$ with the property that each pair of slices, in levels $i,i+1$, and $j,j+1$, are given the same spin.
When $G$ is connected, any such pair of slices form a PSF. Indeed, each vertex in the pair has at least $2$ neighbors in the pair and $1$ neighbor not in the pair. Therefore, a.s. any vertex belongs to a connected subgraph whose boundary is frozen from time $t=0$ on. But this means that there are no infinite connected unfrozen subgraphs. The second part of the theorem follows from
Lemma \ref{lem:finite energy reducing}.
\end{proof}
\begin{rmk}
A similar claim holds for bounded-width deformed cylinder with some condition on the edges between two successive slices. %Moreover, it is enough that for any slice there are slices of some given width in one of the levels above it and one of the levels below it. Even this latter condition may be weakened.
\end{rmk}

Yet, not every cylinder freezes, as the example below shows.
\begin{exam}\label{ex:non freezing cylinder}
Let $G$ be the union of two copies of $K_n$ which share a single common vertex where $n\geq 4$. We now show that the $G$-cylinder does not freeze.
Indeed, suppose that at time $t=0$ the vertices of the five successive slices at levels $-2,-1,0,1,2$ are assigned the following spins:
\begin{itemize}
\item
    All the vertices of slices $-1,-2$ are assigned $-1$.
\item
    All the vertices of slices $+1,+2$ are assigned $-1$.
\item
    All vertices of one copy of $K_n$ in slice $0$ are given spin $-1$
    and the rest are set to $+1$.
\end{itemize}
This happens with positive probability. In this case all vertices except the vertex which is common to the two copies of $K_n$ in slice $0$, will not change their spin. This vertex, however, will change its spin infinitely many times, as it will always have an equal number of $+1$ and $-1$ neighbors.
%!!!!!!!!!!!!!!!!!!!!!!!!!!!!!!!!PICTURE!!!!!!!!!!!!!!!!!!!!!!!!!!!!!!!!!
\begin{figure}[h]
\begin{center}
\includegraphics[width=0.7\textwidth]{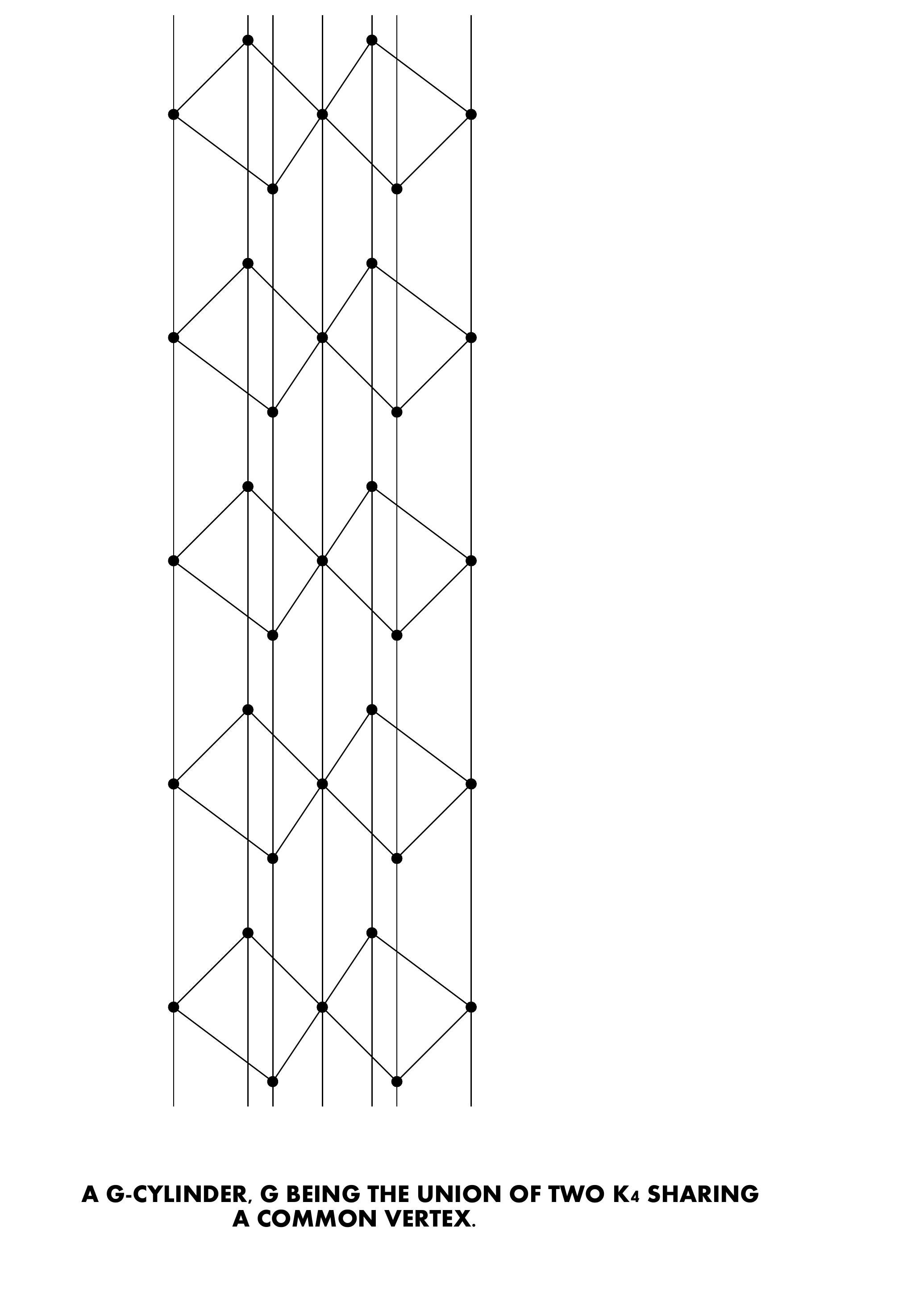}
\caption{G-cylinder for Example \ref{ex:non freezing cylinder}}
\end{center}
\end{figure}
\end{exam}

One the other hand, the dynamics on some cylinders do freeze.
\begin{theorem}
For any $n>1$ the dynamics over the $C_n$-cylinder and the $K_n$-cylinder freeze a.s. Moreover, it freezes in slices with probability $1$.
\end{theorem}
\begin{proof}
As in Theorem \ref{thm:cylinders almost freeze} each vertex lays between two all $+$ or all $-$ pairs of slices which are frozen from time $t=0$.
Consider the subgraph between two such pairs. If it does not freeze in slices, either it does not freeze, or it freezes, but not in slices. In any case, as it is a finite graph, some spin configuration will recur in the sense of \eqref{def:recur1}. If we show that from every such configuration, there is a positive
%\footnote{Oren: Why do you need uniformly bounded away from 0?}
probability to get to a configuration which is frozen in slices in one time unit, then by Lemma \ref{lem:from newman} we will arrive to a contradiction. Indeed, this would imply that with probability one we will eventually get to a configuration which is frozen in slices, which is absorbing, by definition.

We begin with the $C_n$ case. Without loss of generality we shall consider the sub-graph $C_n \times [-2,...,m+2]$ and assume that at some time $t$ the vertices in each of the two pairs of slices: $-1,-2$ and $m+1,m+2$ are either all $+$ or all $-$ (the spins can be different for each pair), while all other vertices have some prescribed spin. The ``procedure'' below produces a sequence of flips which, if occur in order and starting from the above configuration, will result in a configuration where all vertices in each slice have the same spin which is also the same as the spin of all vertices in the slice below or above it. Such configuration is frozen in slices, as desired. This sequence of flips will be feasible, i.e. it could be made to happen via a proper choice for clock-ring times and coin-tosses, which happens with positive probability.

Suppose that the spins of the vertices in slice $-1$ are all $+1$. Consider the slice in level $0$:
\begin{enumerate}
\item
    If all its vertices are $+1$, do nothing and move to consider the slice at level $1$.
\item
    If they are not all $-1$, there is some vertex with a spin of $+1$. Then its
    neighbors in that slice have at least two (out of four) neighbors of the same spin and therefore they can be made $+1$. Iterate to set all vertices in
this slice to $+1$. Move to consider the slice at level $1$.
\item
    If all the vertices have $-1$ spins, consider the slice at level $1$. If all its vertices are $-1$, do nothing and move to consider the slice at level $2$.
\item
Otherwise, there is some vertex with $+1$ spin in the slice of level $1$. Its neighbor at level $0$ has again two out of four neighbors with spin $+1$ and can be made $+1$. Iterate until all the vertices in slice $0$ are $+1$. Then move to consider slice $1$.
\item
    Repeat these steps for the new considered slice. Stop when you reach level $m+1$.
\end{enumerate}
This shows the theorem for $C_n$.

%All the steps above can clearly be translated into a sequence of clock rings %and coins with a positive probability bounded from below by some positive %function of $n,m$.

We now consider $K_n$. For $n=3$, $C_n = K_n$. Therefore we assume $n \geq 4$.
Note that for such $K_n$-cylinder, a slice with all of its vertices having the same spin is frozen, since the majority of the neighbors of each vertex are in the slice itself. Therefore we can now consider a finite subgraph of the form $K_n \times [-1,...,m+1]$ and assume that at some time $t$ the configuration is such that all vertices in slices $-1$, $m+1$ have the same spin (maybe different for each slice) with some prescribed spins for the rest of the vertices. The ``procedure'' for obtaining a sequence of flips that yields a frozen-in-slices configuration is now:
\begin{enumerate}
\item
    Consider slice $0$. If more than half of its vertices have the same spin $s$,
    then any vertex with the opposite spin has at least the same number of
    neighbors with spins $s$ as the number of neighbors with the spin $1-s$.
    Indeed, if $n=2k$, then at least $k+1$ have spins $s$. Any vertex in the
slice, with the opposite spin has at least $k+1$ neighbors with that spin and $2k-1+2 = 2k+1$ neighbors all-together. If $n=2k+1$, then there are at least $k+1$ vertices of spin $s$, and any vertex with the opposite spin has at least that numbers of neighbors with spin $s$ among $2k+1-1+2 = 2k+2$ neighbors all-together. Thus, for any vertex with a spin different from $s$, we flip to $s$. Then we move to consider the next slice.
\item
    If exactly half of the vertices in the slice have spin $s$, then in particular $n=2k$ for some positive integer $k$ and there are $k$ vertices with spin $+1$. Consider a vertex with spin $-1$, then it has at least $k+1$ neighbors with spin $+1$ (including slice $-1$), among $2k+1$ neighbors all-together. Thus, we can set all these vertices to $+1$ as well. We move to consider the next slice.
\item
    We stop when we reach slice $m+1$.
\end{enumerate}
This proves the theorem for $K_n$.
\end{proof}

\begin{rmk}
It follows from \cite{NandaNewmanStein} that on any transitive unimodular graph with $\emph{odd}$ degrees, the dynamics freezes.
\end{rmk}

A sufficient condition on the graph $G$ for the dynamics over $G$-cylinder not to freeze in slices is given in the proposition below.
\begin{prop}
Let $G=(V,E)$ be a graph and suppose that there exists a partition of $V$ into $V_1, V_2$, such that every vertex in $V_1$ has more neighbors in $V_1$ than in $V_2$ and every vertex in $V_2$ has more neighbors in $V_2$ than in $V_1$. Then $\mathbf{P}_{\sigma_{0},\omega}$-a.s. the dynamics on the $G$-cylinder will not freeze in slices.
\end{prop}
\begin{proof}
We follow the construction in Example \ref{ex:non freezing cylinder}. With positive probability slices $-2,-1$ can be initially given the spin $+1$ and the slices $1,2$ can be intially given spin $-1$. Then, still with positive probability, all vertices in slice $0$ with their second coordinate in $V_1$ are initialized with $+1$ and all vertices in slice $0$ with their second coordinate in $V_2$ are initialized with $-1$. This configuration is stable, i.e. there can be no flip which does not $\emph{increase}$ the energy.
\end{proof}
\begin{exam}
Let $G$ be $C_3 \times \mathbb{T}_n$, where $\mathbb{T}_n$ is the $n\times n$-torus or $\mathbb{Z}_n^{2}$, then the $G$-cylinder does not freeze in slices. Note that this is a (unimodular) transitive Cayley graph. Indeed, as in Example \ref{ex:non freezing cylinder} with positive probability the vertices of slices $-2,-1$ can be initially given the spin $+1$ and the vertices of slices $1,2$ can be initially given the spin $-1$. Now the $0$-slice is composed of three tori. With positive probability two of them are set to $+1$ and one is set to $-1$. It is easy to verify that this is a stable state, i.e. no vertex will flip its spin. Thus, with probability $1$ even if the cylinder freezes, it will not freeze in slices.
\end{exam}
%--------------------------------------------------------------------------------------------------------
\newpage
\section{The Loop dynamics and its Applications}\label{sec:dynamics-loops}
$~$\\
\subsection{Background and Preliminaries}\label{sec:loops-background}
$~$\\
In this section we define and explore a different evolution model, which is related to the GSPs of planar spin system. Much like the Glauber dynamics, transitions are local, i.e. allowed only between two configurations which differ at a finite number of vertices. While the Glauber dynamics for Ising yields the Ising Gibbs distribution in equilibrium, under the process which we define here, the weak limits of the distribution of the spins at time $t$ as $t \to \infty$, are all supported on the set of ground states configurations. All the results here are stated for planar lattices, but can be readily extended to all amenable Cayley graphs.

Let $G$ be a planar graph embedded in the plane is such a way that any compact subset of the plane contains only finitely many vertices. Let $\Gamma$ be the set of all simple loops in the dual graph $G^{\ast}$. A set of positive numbers indexed by $\gamma \in \Gamma$, $\{f_\gamma\}_{\gamma\in\Gamma}$, will be called a set of $\emph{loop frequencies}$. An assignment $J_{xy}$ to each edge $x \thicksim y$ will be called a $\emph{coupling}$ or $\emph{interaction}$ as before. As before, $\mathcal{X}$ = $\{-1,1 \}^{G}$ will denote the space of all configurations over $G$, the Hamiltonian of the system is as before
\begin{equation}
\mathcal{H}(\sigma) = -\sum_{x\thicksim y} J_{xy} \sigma_x\sigma_y.
\end{equation}

The $\emph{loop-dynamics}$ with frequencies $\mathbf{f}=\{f_\gamma\}_{\gamma\in\Gamma}$ is a Markov process over the $\mathcal{X}$. With each loop $\gamma$, we associate an independent Poisson clock with rate $f_\gamma$. Given all  couplings $J_{xy}$ and an initial configuration, at every clock ring, say of $\gamma$, the next configuration depends on the total energy contribution from edges whose dual is in $\gamma$, i.e. on
\begin{equation}\label{eq:Hamiltonian-loop}
\mathcal{H}^\gamma_\mathcal{J}(\sigma) = -\sum_{<x,y>\in\gamma^{*}}J_{xy}\sigma_x\sigma_y
\end{equation}
If it is positive, the spins of all vertices within (the finite subset of the plane bounded by) $\gamma$ are flipped. If it is zero, the same spins are flipped with probability $1/2$ independent of everything else. If it is negative, we do nothing. Clearly, in the first case, the total energy of the system is lowered by $2\mathcal{H}^\gamma_\mathcal{J}(\sigma)$ and we call such a step an $\emph{energy-reducing step}$. Note that it is not clear, yet, if this dynamical process is well defined. This will be the first main result of this chapter.

As usual, we denote by $\mathbf{P}_\omega$ the underlying probability measure for all Poisson clocks and coins variables, $\mathbf{P}_{\sigma_{0}}$ for the initial configuration $\sigma_{0}$ and $\mathbf{P}_{\mathcal{J}}$ for
the coupling variables $J_{xy}$. We assume as usual that $\mathbf{P}_{\sigma_{0}}$ is the product measure of some distribution on $\{-1, +1\}$ and in particular, unless stated differently, that this is the symmetric $\pm 1$ distribution. Similarly under $\mathbf{P}_{\mathcal{J}}$ the couplings $J_{xy}$ are i.i.d. and we shall further suppose that their distribution has a first moment.
The joint distribution is $\mathbf{P}_{\mathcal{J},\sigma_{0},\omega}=\mathbf{P}_\mathcal{J}\times\mathbf{P}_{\sigma_{0}}\times\mathbf{P}_\omega$.

For what follows we set $G = \mathbb{Z}^{2}$. We shall say that $v$ is inside $\gamma$ or that $\gamma$ contains $v$ if $v$ is inside the finite subset of the plane which is bounded by $\gamma$. We denote by $V_\gamma$ the set of all vertices which are either inside $\gamma$ or on an edge whose dual is in $\gamma$. Let $l_\gamma$ be the number of dual edges on $\gamma$ and set $S_\gamma = |V_\gamma|$. Two loops $\gamma_1$ and $\gamma_2$ are congruent $\gamma_1 \sim \gamma_2$ if they can be obtained from one another by a composition of translations, rotation by 90 degrees and reflections. The $\emph{type}$ of a loop $[\gamma]$ is the equivalence class to which it belongs, under the above congruency relation. $[\Gamma]$  will stand for the collection of all loop types. %\footnote{Oren: You already have a different definition for $\Gamma$. Either change the one before or from this point on on use $[\Gamma]$ instead of $\Gamma$.}

A set of loop frequencies $f$ is called $\emph{proper}$ if it constant on each $[\gamma]$. In this case, we shall also treat $f$ as a function of $[\Gamma]$. In the same vein, we define $[\Gamma]$-versions of $l$ and $S$ as
$l_{[\gamma]} = l_\gamma$ and $S_{[\gamma]} = S_\gamma$. Finally, $n_{[\gamma]}$ is the number of loops of type $[\gamma]$ which contain the origin. %\footnote{Oren: Isn't $n$ almost equal to $S$, except for the boundary vertices?.}

\subsection{Main Results}\label{sec:loops-main res}
The first thing we show is that the loop-dynamics is well defined, is the rates of the clocks satisfy some inequality, this is Theorem \ref{thm:loop dynamics well defined}. We then prove, in Theorem \ref{thm:weak limits in loop dynamics} that any subsequential limit of the dynamical process is supported only on GSPs. In case of Ferromagnetic system we show that the dynamics has a weak limit, which is supported on the GSP with all the spins being the same. 

\subsection{Properties of the Loop Dynamics}\label{sec:loops-proofs}

The following theorem gives conditions on $f$ such that the loop-dynamics process is well defined
%The first result in this section is that for rates of clocks which satisfy some inequality the dynamics is well %defined.
%We then show that under some more conditions all the weak subsequential limits of this dynamics (for times tending %%to infinity) are measures that are supported on GSPs. We prove some "flexibility" properties for the GSPs in the %support of these measures.
%In the special case of ferromagnetic systems we show that there is actually a single subsequential weak limit.

%\subsection{Weak Limits of Loop Dynamics}\label{sec:loops-proofs}
\begin{theorem}\label{thm:loop dynamics well defined}
If
\begin{equation}
\label{eqn:WellDefinedness}
\sum_{[\gamma]\in \Gamma}n_{[\gamma]} f_{[\gamma]} S_{[\gamma]} < \infty
\end{equation}
then the loop dynamics is $\mathbf{P}_{\mathcal{J},\sigma_{0},\omega}$-a.s well defined.
%\footnote{Oren: I will use ${[\gamma]}$ whenever the property is of the type, not the specific member. I think it's clearer. I.e. $n_{[\gamma]}$ in place of $n_{\gamma}$.}
\end{theorem}
\begin{proof}
To show that the process is well defined it suffices to find $\tau > 0$ such that the spin configuration can be computed up to time $\tau$ given the initial configuration and the ring-times of all clocks and coins. This is clearly the case for a finite graph, but for infinite graph, it is a priori possible to have an infinite sequence of vertices, with the spin of each vertex (up to time $\tau$) depending (due to a possible spin-flip) on the spin of the next vertex in the sequence. The condition in the theorem will ensure that with probability $1$, this dependency sequence is finite for all vertices. We may then compute the configuration at each time $t$, by progressing in units of $\tau$ time.

Let $D$ be the smallest set of vertices, that contains the origin and satisfies the following property:

If $v\in D$ and $\gamma$ is a loop for which $v\in V_\gamma$ and whose clock rang no later than time $\tau$, then $V_\gamma$ is also contained in $D$.
Since equation \eqref{eqn:WellDefinedness} guarantees that for any vertex $v$ and time $t$, the total number of rings up to time $t$ at the clocks of all loops containing $v$ is finite, it is enough to show that $D$ above is almost surely finite.

Indeed, it is not difficult to see that the number of vertices in $D$ is stochastically dominated by the total number of nodes in a Galton-Watson tree, where the distribution of the number of decedents is Poisson with rate
$\tau \sum_{[\gamma]\in \Gamma}n_{[\gamma]} f_{[\gamma]} S_{[\gamma]}$. By choosing $\tau$ small enough, we can make the latter smaller than $1$, in which case the size of the tree is finite almost surely as desired.
\end{proof}
\begin{rmk}
The loop dynamics, as defined, may be considered as an extension of the zero-temperature Glauber dynamics. In the same manner, one may consider the analogous extension of the regular Glauber dynamics (with positive temperature). In that case when the clock of some loop rings, the probability of a flip depends on the relevant Boltzmann-Gibbs law for that temperature. The exact same analysis can be applied to show that this dynamical process is also well-defined.
\end{rmk}
Next, we investigate the possible weak limits of $\mathbf{P}_{\mathcal{J},\sigma_{0},\omega} (\sigma_t \in \cdot)$  as $t \to \infty$. We shall show that any subsequential limit must be supported on ground state configurations.
However, first we need the following.
\begin{lemma}\label{lem:finite loop flips}
Under the above assumptions, if in addition the rates satisfy the inequality (for every edge $e$ and $l\in\mathbb{N}$)
\begin{equation}
\sum_{[\gamma]\in \Gamma_{e}(l)} f_{[\gamma]} < e^{-10 l}
\end{equation}
where $\Gamma_{e}(l)$ is the collection of loops of length at least $l$ which $e$ lays on,
then for any loop $\gamma$ there are only finitely many energy reducing steps. In particular, if the couplings have an absolutely-continuous distribution (w.r.t. Lebesgue measure), then every loop flips only finitely many times.
\end{lemma}
\begin{proof}
Denote by $E_\gamma (t)$ the expectation w.r.t. $\mathbf{P}_{\mathcal{J},\sigma_{0},\omega}$
of the accumulated energy change in $\mathcal{H}^\gamma_\mathcal{J}(\sigma)$ up to time $t$. Note that this
is always non-positive, almost surely finite (due to \eqref{eqn:WellDefinedness} and the first moment assumption) and the same for all loops of the same type.

Fix $N > 0$. For any type $[\gamma]$ there are $cN^2+O(N)$ loops of that type, which are contained in an $N \times N$ square around the origin. Therefore the expected accumulated energy change (in the $N \times N$ square) due to such loops is bounded above by
$(1-\varepsilon)cN^2E_\gamma (t)$ where $\varepsilon = \varepsilon(\gamma, N) > 0$, is a constant that depends on $N$ and on the structure of the loop $\gamma$, and which tends to $0$ as $N\rightarrow \infty$.

For such a square, the total energy change up to time $t$ (by which we mean the difference in the Hamiltonians restricted only to the square and its boundary between time $0$ and time $t$) is due to two sources: A negative contribution from spin-flips inside loops which are strictly contained in the square and (negative or positive) contribution from spin-flips inside loops which intersect the boundary of the square - both up to time $t$.
We claim that there is a constant $L$ such the latter is bounded from above, in expectation, by $tLN\log (N)$, for large enough $N$.
Indeed, it is enough to bound the expected total number triplets $(e,\gamma, t_i)$ such that $\gamma$ is a loop which intersects the boundary of the square, $e$ is an edge in the square which lays on $\gamma$ and the clock of $\gamma$ rang at time $t_i \leq t$.
This number is clearly bounded from above by a constant multiple of $t \sum_{0\leq l \leq \frac{N}{2}}(N-l)\sum_{[\gamma]\in \Gamma_{e}(l)} f_{[\gamma]}$, since if a loop contains an edge of distance $k$ to the boundary, and also touches the boundary - its length must be at least $k$ (actually at least $2 k$). 
The last sum is bounded from above by $t (N\log (N) + O(N^{-9}))$, as can be seen if we divide the sum into summation over $l < \log (N)$, and summation over $\log (N) \leq l$. Hence the total contribution of the boundary is no more than $tLN\log (N)$.
%Indeed, loops of length less than $\log (N)$, which are not contained in the square, but touch the boundary but  make for, in expectation, a linear contribution (both in $N$ and $t$)
%and longer loops, although their number is infinite, have such a slow clock rates that their expected contribution is also bounded, by some multiple of $Nt$ (the calculation is similar to that in Theorem \ref{thm:loop dynamics well defined}). 

On the other hand, the expected total energy change in the box is bounded by the difference between the minimal and maximal total energy inside the square, which is at least $-MN^2$ where $M$ is some positive constant.
Putting the above observations together gives
%\begin{equation}
%(1-\varepsilon)cN^2 E_\gamma (t) \geq tNL-MN^2
%\end{equation}
\begin{equation}
(1-\varepsilon)cN^2 E_\gamma (t) + tN\log (N)L \geq -MN^2
\end{equation}
Contributions from other internal loops are ignored as they can only sharpen the inequality (they are a non-positive amount which is added to the left hand side).
%\footnote{Oren:??}
Dividing by $N^2$ and letting $N\rightarrow\infty$, and then $\varepsilon\rightarrow 0$, keeping $t$ fixed, we obtain $E_\gamma (t) \geq -M$. Then, letting $t\rightarrow\infty$ and using the monotonicity of $E_\gamma (t)$, we get that the total energy
change which is made by any loop $\gamma$ is finite almost surely.

Finally, given a specific loop $\gamma$, if we assume that the couplings have an absolutely-continuous distribution w.r.t. Lebesgue measure, then there is no spin choice $\sigma$ for which $\mathcal{H}^\gamma_\mathcal{J}(\sigma) = 0$ (no zero-energy-changes can occur).
Hence, each energy reducing spin-flip decreases the energy by at least some positive constant. Thus we obtain that the number of spin-flips of $\gamma$ (spin flips inside $\gamma$ which are caused by clock rings of $\gamma$) is finite almost surely.

%The last part of the lemma follows since in the absolutely-continuous case, no zero-energy-changes can occur.
\end{proof}
\begin{rmk}
Actually the additional inequality for the rates was added only for simplicity, and is unnecessary, as the contribution of the boundary is also expected to be nonpositive. 
\end{rmk}
\begin{lemma}\label{lem:no loop is negative afterwhile}
Under the above assumptions, if in addition the rate of every loop is positive, then for any loop there is some time $t$ such that after that time the value of the Hamiltonian over the edges dual to that loop's edges is non-positive.
\end{lemma}
\begin{proof}
$\mathbf{Sketch.}$ 
%\footnote{Oren: Didn't check thoroughly. Can't you just argue that there's a positive probability for the clock of the loop to ring before any clock of a loop that contain a vertex of an edge whose dual is in the loop - because of summability of the rates and then use the newman lemma below.}
Consider a fixed loop $\gamma$, and take a large enough box around it. After some time no loop which is contained in that box makes an energy-reducing step. But it still may happen that that spin-flips inside larger loops which are not contained in the box will change the Hamiltonian of $\gamma$.
We would like to show that after some time those changes will not make that Hamiltonian positive.
The idea is that the total possible flip rate of long enough loops
(such as loops that intersect the large box, yet not contained in it), must be smaller than that of $\gamma$ as large loops have extremely slow rates,
and large enough loops with a common edge have small total flip rate, due to convergence. We would like to show that if there were infinitely many such
changes, $\gamma$, whose clock has much more often rings, must make an energy-reducing step again. This follows immediately from Lemma \ref{lem:from newman}.
Applying the lemma for $A$ as being all the configurations that when reducing them to the large box around $\gamma$ leaves the Hamiltonian of $\gamma$ positive,
and $B$ be the event that $\gamma$ flips (given any initial setting and in at most one time unit), gives us the result.
\end{proof}

We can now state:
\begin{theorem}\label{thm:weak limits in loop dynamics}
Any subsequential weak limit of $\sigma_{t}$ (for a sequence of times which tend to infinity) under $\mathbf{P}_{\mathcal{J},\sigma_{0},\omega}$ must be a
translation invariant measure, supported only on ground state configurations.
\end{theorem}
\begin{proof}
Translation invariance is clear. By the above lemma, for any loop $\gamma$ we have
$\mathbf{P}_{\mathcal{J},\sigma_{0},\omega} (\mathcal{H}^\gamma_\mathcal{J}(\sigma_t) > 0) \to 0$ as $t \to \infty$. Therefore in any subsequential limit $\sigma_{\infty}$ we have $\mathcal{H}^\gamma_\mathcal{J}(\sigma_{\infty}) \leq 0$ for all $\gamma$ a.s.. Therefore $\sigma_{\infty}$ is a ground state configuration.
\end{proof}

Next we show that if $J_{xy} = 1$ for all $x \thicksim y$ with probability $1$, then under the loop dynamics there is a weak limit to $\sigma_t$ as $t\rightarrow\infty$ and it is supported on two configurations, those with all spins $+1$ or all spins $-1$. 
Yet, vertices flip their spin infinitely many times, i.e. there is no strong pointwise limit for $\sigma_t$. The reason is that having a strong limit is a translation-invariant event, which may be approximated by knowing the initial spins, in some neighborhood of the origin and by knowing the rings times and coins for loops in that neighborhood, up to some time. Thus this is a $0-1$ event. In the same manner, freezing to a configuration where all the spins are $+1$, is also a $0-1$ event. But now symmetry (with the other spin) shows it must be a $0$ event.
% \footnote{Oren: Where do you show this?}
It should be noted that such a result is still open for zero-temperature Glauber dynamics in the square lattice and an understanding of the support of the corresponding subsequential weak limit is lacking. See (\cite{CamiaSantisNewman}).

\begin{claim}\label{clm:weak limit in loops +1}
If $\mathcal{J} \equiv 1$, then $\sigma_t \Longrightarrow \frac{1}{2}\delta_+ + \frac{1}{2}\delta_-$ where $\delta_+$ ($\delta_-$) is the Dirac mass on the all $+$ ($-$) configuration. Hence, for almost every $\sigma_0$, modulo the global flip, we have convergence in probability to the dirac mass on the GSP of a single spin (all $+$ or all $-$).
%
%a.s. with respect to $\mathbf{P}_\omega$ and for almost every neighborhood of the GSP, the set of times for which $\sigma_t$ is in the neighborhood is of density $1$.
%Furthermore the statement also holds with $\mathbf{P}_{\sigma_{0}}$-probability $1$, if one conditions on the initial configuration $\sigma_{0}$.
\end{claim}
\begin{proof}
It follows from Theorem \ref{thm:weak limits in loop dynamics} that subsequential weak limits for $\sigma_{t}$ are
translation invariant measures supported on ground states. Theorem \ref{thm:single GSP in Z^d if J=1} shows that  the support of these measures contains only monochromatic configurations, i.e. any subsequential weak limit must be of form $p\delta_+ + (1-p)\delta_-$ for some $p \in [0,1]$ which depends on the given sequence. Invariance under $\mathbf{P}_{\mathcal{J},\sigma_{0},\omega}$ of the distribution of $\sigma_t$ w.r.t. to global spin-flip implies that $p$ must be $1/2$. 
%Moreover, conditioned on the initial configuration $\sigma_{0}$, it is still true that subsequential weak limits are of the form $p\delta_+ + (1-p)\delta_-$, where $p$ may depend on the $\sigma_{0}$. Since $p$ is translation invariance, by Ergodicity it must be a constant and the constant must be $1/2$.
\end{proof}

\begin{rmk}
$ $ \\
\begin{enumerate}
\item
    The above analysis carries through for any product measure for the couplings, as long as their distribution is strictly positive. These models are called $\emph{Ferromagnetic}$.
\item
    The loop dynamics can be generalized to lattices of higher dimensions, where instead of loops we consider closed surfaces of co-dimension $1$. Similar bounds for the rates are needed to make the process well defined and all statements carry through with minor modifications (in particular the case $J=+1$). In fact, everything can be generalized to amenable Cayley graphs.
\end{enumerate}
\end{rmk}

We end this section with a result concerning the ``stability'' of a sequential limit under small perturbations of the couplings. Roughly speaking we show that there exist couplings, such that changing their value a small enough amount, without changing the order of clocks rings does not change evolution of spin configurations. We assume henceforth that the couplings are chosen from a product measure with marginals which are supported on the entire real line and with finite first moment.
\begin{proposition}\label{prop:walls between weak limits in dynamics}
Under the above assumptions, let $\sigma_1,\sigma_2$ be two spin configurations which are in the support of some subsequential weak limit. Consider the subgraph consisting of all edges which are dual to edges $(u,v)$ where $\sigma_1(u) = \sigma_2(u)$ but $\sigma_1(v) = -\sigma_2(v)$, then this subgraph is a collection of two-sided infinite dual paths.
\end{proposition}
\begin{proof}
As the subsequential limits are translation invariant measures, this proposition is a direct result of \cite{Newman00natureof}.
\end{proof}
\begin{rmk}
These dual paths are called $\emph{domain walls}$.
Note that dual edges whose primal edges are $\emph{fixed}$ (according to the terminology of Section \ref{sec:spinglasses}) will be in no domain wall, and hence will flip only finitely many times.
\end{rmk}

\begin{claim}\label{clm:flexibility out of walls}
Given the couplings, the initial spin configuration, and the ring times, almost surely there exists an edge $e$ and a positive number $\varepsilon$ such that changing the coupling of $e$ by any number whose absolute value is smaller than $\varepsilon$ does not change $\{\sigma_t \}_{t\geq 0}$.
\end{claim}
\begin{proof}
%\footnote{Oren: Didn't check below.}
$\mathbf{Sketch.}$ Consider any fixed edge $e$. In any weak limit it must be satisfied. Thus, there must be a time after which both of its spins flip together (or maybe do not flip at all). Up to that time, $t_0$, for only finitely many dual loops which intersect it (and hence - contain its dual edge) - their Poisson clock had ringed (we count multiplicity). Due to continuity, there exist some positive $\delta$ such that changing the value of this edge's interaction by not more than $\delta$ would not change the sign of any of the dual loops whose clock rang (at the corresponding times). Note that if we increase the absolute value of the interaction by that $\delta$, we do not change the process up to time $t_0$. After time $t_0$, as our edge is always satisfied, we increased the value of any dual loop which intersects it, and hence we clearly did not cause any new flip. Thus, in this case the process does not change.
\\ What about lowering the absolute value of the interaction? 
Assume that we fix all the interactions other then that of $e$, and all the clock times (and coins, if needed), and we condition on $e$ being a fixed edge. From the above it follows that a.s. there are at most countably many distinct options for the series of loop flips. Indeed, for any value of the interaction of $e$ (which is conditioned to be a fixed edge), there is some interval with positive length containing this value (maybe as boundary), s.t. changing the value of $e$ within that interval does not change the evolution of flips. If $x,y$ are two possible interaction values for $e$, for which it is a fixed edge, either the interiors of their corresponding intervals are distinct or they coincide. Moreover the flip series depends on the interval and not on the interaction itself. As every interval contains a rational number - there are at most countably many such intervals.
Thus, to any sequence of ring times, and interactions of all edges except $e$ there corresponds a unique a.s. countable set of interval boundaries. For any value of interaction other then these boundaries (which is in the half line of values where $e$ is fixed) there is some $\varepsilon > 0$ such that changing the value of the interaction in no more than $\varepsilon$ (to either side) leaves the flip sequence unchanged. As the set of boundaries is of measure $0$ we are done.
%given the rest of the interactions (apart from the one of the edge we are concerned with), the clock rings sequences and the information of whether or not a flip occurs in any loop, we obtain an infinite sequence of linear inequalities which the interaction must fulfil in order that the flips will be exactly the predicted flips. It is easily verified that these inequalities leave us with an interval of possible values to the unknown interaction. This interval may be empty. But its two boundary values are functions of the above information and do not depend on the unknown interaction. Thus, if we know that the interval is not empty (as we know a.s. in the case of a fixed edge, due to what we have shown), due to continuity of the distribution, almost surely the unknown interaction, conditioned to be in the interval (and to be fixed), does not lay in the boundary on the interval. Hence there is some positive $\varepsilon$ such that changing it leaves the process unchanged.
\end{proof}
%IN THE FUTURE- TO ADD MORE FLEXIBILITY RESULTS????????????????!!!!!!!!!!!!!!!!!!!!!!!!!!!
\newpage

\bibliography{bibli}
\bibliographystyle{alpha}

% *************************************************************************************************************************************
% *************************************************************************************************************************************
% *************************************************************************************************************************************

\end{document}